\documentclass[journal,twoside,web]{ieeecolor}
\usepackage{generic}

\usepackage{cite}
\usepackage{amsmath,amssymb,amsfonts}
\usepackage{graphicx}
\usepackage{hyperref}
\usepackage{textcomp}
\def\BibTeX{{\rm B\kern-.05em{\sc i\kern-.025em b}\kern-.08em
    T\kern-.1667em\lower.7ex\hbox{E}\kern-.125emX}}
\usepackage{graphicx}
\usepackage{epsfig}
\usepackage{times}
\usepackage{amsmath}
\usepackage{thmtools,thm-restate}
\usepackage{amssymb}
\usepackage[section]{placeins}
\usepackage{placeins}
\usepackage{amsmath}
\usepackage{multirow}
\usepackage{graphicx}
\usepackage[normalem]{ulem}
\usepackage{subcaption}
\usepackage{algorithm,tabularx}
\usepackage{algpseudocode}
\usepackage{amsfonts}
\usepackage{cases}
\usepackage{romannum}
\usepackage{textcomp}
\usepackage{textcomp}
\providecommand{\U}[1]{\protect\rule{.1in}{.1in}}
\providecommand{\U}[1]{\protect\rule{.1in}{.1in}}

\newtheorem{assumption}{Assumption}
\newtheorem{theorem}{Theorem}
\newtheorem{corollary}{Corollary}

\newtheorem{proposition}{Proposition}
\newtheorem{remark}{Remark}
\newtheorem{definition}{Definition}

\useunder{\uline}{\ul}{}
\makeatletter
\newcommand{\multiline}[1]{  \begin{tabularx}{\dimexpr\linewidth-\ALG@thistlm}[t]{@{}X@{}}
#1
\end{tabularx}
}
\makeatother
\usepackage{cite}

\usepackage{centernot}
\let\labelindent\relax
\usepackage{enumitem}
\setlist[itemize]{leftmargin=*}
\usepackage{makecell}
\usepackage[normalem]{ulem}
\usepackage{cuted}
\usepackage{hyperref}
\usepackage{stfloats}
\usepackage{lscape}

\newcommand{\R}{\mathbb{R}}
\newcommand{\N}{\mathbb{N}}

\newcommand{\T}{\top}

\newcommand{\I}{\mathbf{I}}
\newcommand{\0}{\mathbf{0}}
\newcommand{\E}{\mathcal{E}}

\newcommand{\X}{\mathcal{X}}
\newcommand{\diag}{\text{diag}}
\newcommand{\tsup}[1]{\textsuperscript{#1}}
\newcommand{\mb}[1]{\mathbf{#1}}
\renewcommand{\H}{\mathcal{H}}

\newcommand{\bm}[1]{\begin{bmatrix}#1\end{bmatrix}}

\algdef{SE}[SUBALG]{Indent}{EndIndent}{}{\algorithmicend\ }%
\algtext*{Indent}
\algtext*{EndIndent}

\begin{document}

\title{Inventory Consensus Control in Supply Chain Networks using Dissipativity-Based Control and Topology Co-Design
}

\author{Shirantha Welikala, \IEEEmembership{Member, IEEE}, Hai Lin \IEEEmembership{Senior Member, IEEE}, and Panos J. Antsaklis \IEEEmembership{Fellow, IEEE} 
\thanks{The support of the National Science Foundation (Grant No. CNS-1830335, IIS-2007949) is gratefully acknowledged.}
\thanks{Shirantha Welikala is with the Department of Electrical and Computer Engineering, School of Engineering and Science, Stevens Institute of Technology, Hoboken, New Jersey, USA 07030 \texttt{{\small \{swelikal\}@stevens.edu}}. Hai Lin and Panos J Antsaklis are with the Department of Electrical Engineering, College of Engineering, University of Notre Dame, Notre Dame, Indiana, USA 46556, \texttt{{\small \{hlin1,pantsakl\}@nd.edu}}.}}

\maketitle

\pagenumbering{arabic}
\thispagestyle{plain}
\pagestyle{plain}

\begin{abstract}
Recent global and local phenomena have exposed vulnerabilities in critical supply chain networks (SCNs), drawing significant attention from researchers across various fields. Typically, SCNs are viewed as static entities regularly optimized to maintain their optimal operation. However, the dynamic nature of SCNs and their associated uncertainties have motivated researchers to treat SCNs as dynamic networked systems requiring robust control techniques. In this paper, we address the \emph{SCN inventory consensus} problem, which aims to synchronize multiple parallel supply chains, enhancing coordination and robustness of the overall SCN. To achieve this, we take a novel approach exploiting dissipativity theory. In particular, we propose a dissipativity-based co-design strategy for distributed consensus controllers and communication topology in SCNs. It requires only the dissipativity information of the individual supply chains and involves solving a set of convex optimization problems, thus contributing to scalability, compositionality, and computational efficiency. Moreover, it optimizes the robustness of the SCN to various associated uncertainties, mitigating both \emph{bullwhip} and \emph{ripple} effects. We demonstrate our contributions using numerical examples, mainly by comparing the consensus performance with respect to standard steady-state control, feedback control, and consensus control strategies. 
\end{abstract}


\section{Introduction}\label{Sec:Introduction}

Global phenomena like pandemics, conflicts, climate change effects, and technological advancements have caused various supply chain issues around the world \cite{Aday2022,Jagtap2022,Leslie2022,Rajaeifar2022}. Moreover, local phenomena like holidays, shopping sales, trade union actions, adverse weather conditions, infrastructure failures, and economic crises always exacerbate such supply chain issues \cite{Jiang2023}. These challenges highlight the need to carefully design and control supply chain networks (SCNs). Consequently, SCNs have garnered significant attention from researchers across different fields in recent years \cite{Taboada2022}.

Traditional operations research-based supply chain management (SCM) systems assume static, deterministic SCNs and focus on limited objectives. Some such SCM system examples include approaches like mathematical programming \cite{Mula2010}, inventory-transportation \cite{Qui1999}, and inventory-routing \cite{Andersson2010}. In contrast, control theoretic SCM techniques are gaining popularity because they account for the dynamic nature and uncertainties of SCNs, leading to adaptive, robust, and resilient SCNs. Extensive discussions on the advantages of control theoretic SCM systems can be found in the review papers \cite{Surana2005,Ivanov2012,Ivanov2018,Janamanchi2013,Monostori2015,Taboada2022}.

Typical SCN challenges arise from uncertainties in the associated supplies \cite{Chen2022}, logistics \cite{Liu2015}, inventory \cite{Raj2022}, product lifetimes \cite{Sun2006} and demands \cite{Li2020}. To counter such challenges, robust control mechanisms must be implemented at different stages of the SCN \cite{Li2020,Liu2015}. On the other hand, SCM systems are typically designed to optimize associated operational costs  \cite{Wang2021b}, error metrics (e.g., consensus \cite{Li2020} and tracking \cite{Anne2009}), time metrics (e.g., lead time \cite{Hosoda2015} and waiting time \cite{Welikala2020J2}) or robustness measures (e.g., \emph{bullwhip effect} \cite{Lee1997,Chen2000} or \emph{ripple effect} \cite{Ivanov2015}). Among these objectives, when designing robust control mechanisms for SCNs, particular focus is given to minimizing the bullwhip and ripple effects.

The bullwhip effect is the amplification of the demand uncertainty as one moves up the supply chain (towards the manufacturer/supplier) \cite{Lee1997}. It can be viewed as a byproduct of the SCM system trying to overcompensate for demand uncertainty, order batching, price fluctuations, and rationing and shortage gaming \cite{Lee1997}. An extensive body of literature has demonstrated the existence, identified the probable causes, and proposed techniques to reduce the bullwhip effect \cite{Chen2000}. Some such techniques are \cite{Lee1997}: (1) avoiding frequent demand forecast updates, (2) breaking batched orders, (3) stabilizing prices, and (4) eliminating gaming in shortage situations. Lately, more quantifiable and control theoretic approaches to formally study and minimize this bullwhip effect have come to prominence \cite{Li2020,Li2018,Chen2000,Dejonckheere2003,Huang2007,Papanagnou2022}. The work in \cite{Chen2000} studies a simple two-stage manufacturer-retailer model using a moving average demand estimator, concluding that centralizing demand information can minimize the bullwhip effect. Conversely, \cite{Dejonckheere2003} demonstrates that standard order-up-to policies inherently generate the bullwhip effect under established demand forecasting methods and proposes a dynamic replenishment rule based on smoothed demand forecasts as a remedy. The work in \cite{Wang2016Bullwhip} provides a comprehensive review of the bullwhip effect and discusses various strategies to mitigate it, including information sharing and supply chain coordination.

On the other hand, the ripple effect is the amplification of the changes in some critical parameters (or disruptions in some structures) through the rest of the supply chain in all directions \cite{Ivanov2014a}. A detailed comparison between the bullwhip effect and the ripple effect can be found in \cite[Tab. 1]{Ivanov2015}. As an application example, \cite{Ivanov2015} addresses a supply chain design problem that considers ripple effects from disruptions by optimizing multiple prioritized objectives, including performance indicators for stability, reliability, robustness, and resilience. This approach leverages the control theoretic description of SCNs as controllable dynamic systems with dynamic structures, first proposed in \cite{Ivanov2010}. The work in \cite{Dolgui2018} provides a comprehensive analysis and review of the ripple effect in SCNs while emphasizing robust control strategies as a means to mitigate the ripple effect.

Building upon the preceding discussion, it is evident that there is a significant need for control-theoretic SCM systems and supply chain coordination schemes that leverage information-sharing via communication topologies to mitigate adverse effects like the bullwhip and ripple effects. Addressing this gap, in this paper, we focus on the inventory consensus control problem, aiming to synchronize multiple parallel supply chains within an SCN. In particular, we propose a dissipativity-based consensus control and communication topology co-design scheme that both optimizes the communication topology and minimizes the bullwhip and ripple effects in a control-theoretic sense. The resulting co-design ensures a robust synchronization of the supply chains, thereby enhancing the overall SCN resilience and efficiency.

Several studies have addressed the consensus control problem in SCNs using control-theoretic methods. For instance, \cite{Sun2022Dynamical} proposes distributed consensus tracking control techniques for variable-order fractional SCNs, employing multi-agent neural network-based methods to achieve synchronization among supply chains. Similarly, \cite{Li2020} investigates the $H_\infty$ consensus problem for SCNs under switching topologies and uncertain demands, designing switching controllers that ensure consensus while attenuating the bullwhip effect. Additionally, finite-time consensus tracking control approaches have been developed to handle disturbances, uncertainties, and actuator faults in chaotic SCNs, utilizing super-twisting algorithms in \cite{Liu2022Distributed}. The work in \cite{Fanti2012} employs distributed consensus algorithms for task allocation in SCM, formulating the problem as a discrete consensus process to minimize task costs and improve coordination among supply chain actors. In manufacturing processes, distributed control of stock levels using consensus algorithms has been proposed, demonstrating global stability of buffer stock levels at consensus values through probabilistic methods \cite{Tsumura2019}.

\paragraph*{\textbf{Contributions}} 
Despite these promising developments, existing distributed control solutions in SCNs assume the underlying communication topology as prespecified/fixed, thereby leaving a crucial avenue to mitigate bullwhip and ripple effects unexplored. 
Motivated by this idea, in this paper, we propose to co-design the distributed controllers and the communication topology to ensure SCN consensus and optimize communication topology while also minimizing bullwhip and ripple effects. 
In particular, we propose a dissipativity-based control solution that holistically addresses these challenges in an energy-efficient manner.  
To the best of our knowledge, dissipativity theory, which is well known for energy-efficient control solutions \cite{Kottenstette2014}, has not been applied before to control problems in SCNs. 
We exploit the inherent dissipativity properties of supply chains to efficiently design large-scale SCNs, leveraging the compositionality of dissipativity \cite{Arcak2022}. 
By formulating the problem using linear matrix inequalities (LMIs), we enable computationally efficient co-design of distributed controllers and communication topology \cite{Lofberg2004,Boyd1994}. Additionally, our dissipativity- and LMI-based approach offers scalability, decentralizability, and compositionality - enhancing the overall solution's practicality.
To showcase the effectiveness of the proposed control solution, we implemented a generic SCN simulator\footnote{Publicly available at: \href{https://github.com/shiran27/SupplyChainSimulator}{https://github.com/shiran27/SupplyChainSimulator}} 
along with several standard control solutions to provide an extensive comparison study under diverse challenging scenarios.

\paragraph*{\textbf{Organization}}
This paper is organized as follows. In Section \ref{Sec:Preliminaries}, we provide several preliminary concepts and results related to dissipativity, networked systems, and topology design. Section \ref{Sec:Consensus} formulates the inventory consensus control problem in SCNs. In Section \ref{Sec:ConsensusSolution} we provide the proposed dissipativity-based control and topology co-design solution. Finally, Section \ref{Sec:SimulationResults} discusses several simulation results before concluding the paper in Sec. \ref{Sec:Conclusion}.

\paragraph*{\textbf{Notation}}
The sets of real and natural numbers are denoted by $\R$ and $\N$, respectively. An $n$-dimensional real vector is denoted by $\R^n$. We define $\N_N\triangleq\{1,2,\ldots,N\}$ where $N\in\N$.  An $n\times m$ block matrix $A$ is represented as $A=[A_{ij}]_{i\in\N_n, j\in\N_m}$ where $A_{ij}$ is the $(i,j)$\tsup{th} block of $A$. $[A_{ij}]_{j\in \N_m}$ represents a block row matrix and $\diag([A_{ii}]_{i\in\N_n})$ represents a block diagonal matrix. Note that, all the indexing-related subscripts may also be written as superscripts. The transpose of a matrix $A$ is denoted by $A^\T$ and $(A^\T)^{-1} = A^{-\T}$. The zero and identity matrices are denoted by $\0$ and $\I$, respectively (dimensions will be clear from the context). A symmetric positive definite (semi-definite) matrix $A\in\R^{n\times n}$ is represented as $A=A^\T>0$ ($A=A^\T \geq 0$). Unless stated otherwise, $A>0 \iff A=A^\T>0$. 
The symbol $\bigstar$ is used to represent conjugate blocks in symmetric matrices, and $\H(A) \triangleq A+A^\T$. 
Given sets $A$ and $B$, $A\backslash B$ indicates the set subtraction operation. 
$\mb{1}_{\{\cdot\}}$ denotes the indicator function, and   
$\mb{1}_n$ denotes the vector of ones in $\R^n$. $\text{randi}(x,y)$ denotes a uniformly distributed random integer between $x$ and $y$, and $\mathcal{N}(\mu,\sigma)$ denotes the random normal distribution with mean $\mu$ and standard deviation $\sigma$.

\section{Preliminaries}\label{Sec:Preliminaries}

\subsection{Dissipativity}

Consider the discrete-time dynamical system  
\begin{equation}\label{Eq:GeneralSystem}
\begin{aligned}
    x(t+1) = f(x(t),u(t)),\\
    y(t) = h(x(t),u(t)),
    \end{aligned}
\end{equation}
where $x(t)\in\R^{n_x}$, $u(t)\in \R^{n_u}$, $y(t)\in\R^{n_y}$ and $t\in\N$ are the state, input, output and time, respectively. Suppose there exists a set of equilibrium states $\X\subset \R^{n_x}$ where for every $x^*\in\X$ there is a unique $u^*\in\R^{n_u}$ satisfying $x^* = f(x^*,u^*)$, and both $u^*$ and $y^*\triangleq h(x^*,u^*)$ are implicit functions of $x^*$.

To analyze the dissipativity properties of \eqref{Eq:GeneralSystem} without explicit knowledge of its equilibrium points $\X$, we recall the \emph{equilibrium-independent dissipativity} (EID) concept \cite{Arcak2022}. Note that this EID concept includes the conventional dissipativity property \cite{Willems1972a} when $\X=\{\0\}$ and for $x^*=\0 \in\X$, the corresponding $u^*=\0$ and $y^*=\0$.

\begin{definition}\label{Def:EID}
The system \eqref{Eq:GeneralSystem} is EID (from input $u$ to output $y$) under the supply rate function $s:\R^{n_u}\times\R^{n_y}\rightarrow \R$ if there exists a storage function $V: \R^{n_x} \times \X \rightarrow \R$ such that $V(x^*,x^*)=0$, $V(x(t),x^*)>0$ and 
$$
V(x(t+1),x^*) - V(x(t),x^*) 
\leq  s(u(t)-u^*,y(t)-y^*),
$$
for all $(u(t),t,x(1),x^*) \in \R^{n_u}\times\N\times\R^{n_x}\times \X$.
\end{definition}

Based on the used supply rate function $s(\cdot,\cdot)$, the above-defined EID property can be specialized. Following \cite{WelikalaP52022}, we use a quadratic supply rate function determined by a coefficients matrix $X=X^\T \in\R^{(n_u+n_y)\times(n_u+n_y)}$, to define a specialized EID property we named the $X$-EID property.

\begin{definition}\label{Def:X-EID}
The system \eqref{Eq:GeneralSystem} is $X$-EID where $X = X^\T \triangleq [X^{kl}]_{k,l\in\N_2}$ if it is EID under the supply rate function:
\begin{equation*}
    s(u-u^*, y-y^*) \triangleq 
    \bm{u-u^* \\ y-y^*}^\T 
    \bm{X^{11} & X^{12}\\X^{21} & X^{22}}
    \bm{u-u^*\\y-y^*}.
\end{equation*}
\end{definition}

Moreover, based on the choice of the coefficients matrix $X$, the above-defined $X$-EID property can represent several properties of interest as given in the following remark. 

\begin{remark}\label{Rm:X-DissipativityVersions}
If the system \eqref{Eq:GeneralSystem} is $X$-EID with:  
\begin{enumerate}
\item $X = \bm{\0 & \frac{1}{2}\I \\ \frac{1}{2}\I & \0}$, then it is \emph{passive};
\item $X = \bm{-\nu\I & \frac{1}{2}\I \\ \frac{1}{2}\I & -\rho\I}$, then it is \emph{strictly passive} ($\nu$ and $\rho$ respectively are input feedforward and output feedback passivity indices) \cite{WelikalaP42022};
\item $X = \bm{\gamma^2\I & \0 \\ \0 & -\I}$, then it is finite-gain \emph{$L_2$-stable} ($\gamma$ is the $L_2$-gain) \cite{WelikalaP42022};
\item $X = \bm{-a\I & \frac{a+b}{2b}\I \\ \frac{a+b}{2b}\I & -\frac{1}{b}\I}$ (or $X = \bm{-ab\I & \frac{a+b}{2}\I \\ \frac{a+b}{2}\I & -\I}$) with $b>a$ (or $b>a$ and $b>0$), then it is sector bounded ($a$ and $b$ are sector bound parameters) \cite{Kottenstette2014};
\end{enumerate}
in an equilibrium-independent manner. 
For notational convenience, the above Cases 2 and 3 are also denoted as the system \eqref{Eq:GeneralSystem} being IF-OFP($\nu,\rho$) and L2G($\gamma$), respectively.  
\end{remark}

A necessary and sufficient condition for \eqref{Eq:GeneralSystem} to be $X$-EID can be found in the form of a linear matrix inequality (LMI) problem when \eqref{Eq:GeneralSystem} is linear and time-invariant.

\begin{proposition}\label{Pr:LTISystemXDisspativity}
The linear time-invariant (LTI) system
\begin{equation}\label{Eq:Pr:LTISystemXDisspativity1}
\begin{aligned}
    x(t+1) =&\ A x(t) + B u(t),\\
    y(t) =&\ Cx(t) +Du(t),
\end{aligned}
\end{equation}
is $X$-EID (from input $u$ to output $y$) if there exists a matrix $P\in\R^{n_x\times n_x}$ such that $P>0$ and 
\begin{equation}
\nonumber
\begin{aligned}
\label{Eq:Pr:LTISystemXDisspativity2}
\scriptsize
\bm{
P & PA & PB \\
A^\T P & P + C^\T X^{22} C & C^\T X^{21} + C^\T X^{22} D\\
B^\T P & X^{12} C + D^\T X^{22} C & X^{11} + D^\T X^{21} + X^{12} D + D^\T X^{22}D
}
\\ \normalsize
\geq 0.
\end{aligned}
\end{equation}
\end{proposition}
\begin{proof}
See Appendix \ref{App:Pr:LTISystemXDisspativity}.
\end{proof}

The following corollary considers a particular LTI system that includes a local feedback controller $u(t) = Lx(t)$ and provides an LMI problem for synthesizing the local controller to enforce/optimize the $X$-EID property. 

\begin{corollary}\label{Co:LTISystemXDisspativation}
The discrte-time LTI system
\begin{equation}\label{Eq:Co:LTISystemXDisspativation1}
\begin{aligned}
x(t+1) =&\ (A+BL)x(t) + \eta(t),\\
y(t) =&\ Cx(t),   
\end{aligned}
\end{equation}
where $C = [\I \ \0]$ is $X$-EID with $X^{22}<0$ from external input $\eta(t)$ to output $y(t)$ if and only if there exists $P>0$ and $K$ such that 
\begin{equation}
\scriptsize 
\nonumber 
\label{Eq:Co:LTISystemXDisspativation2}
\bm{ 
(-X^{22})^{-1} & \0 & CP & \0 \\
\0 & P & AP+BK & \I \\
P C^\T & PA^\T + K^\T B^\T & P  & PC^\top X^{21}\\
\0 & \I & X^{12} C P & X^{11}} \normalsize  \geq 0
\end{equation}  
and $L=KP^{-1}$.
\end{corollary}
\begin{proof}
See Appendix \ref{App:Co:LTISystemXDisspativation}.
\end{proof}

\subsection{Networked Systems}

Let us consider the networked system $\Sigma$ shown in Fig. \ref{Fig:InterconnectedSystems} comprised of $N$ independent discrete-time dynamic subsystems $\Sigma_i,i\in\N_N$ and a static interconnection matrix $M$ that defines how the subsystem inputs and outputs, an exogenous input $w(t)\in\R^{n_w}$ (e.g., disturbance) and an interested output $z(t)\in\R^{n_z}$ (e.g., performance) are interconnected. 

\begin{figure}[!t]
		\centering
		\captionsetup{justification=centering}
		\includegraphics[width=1.7in]{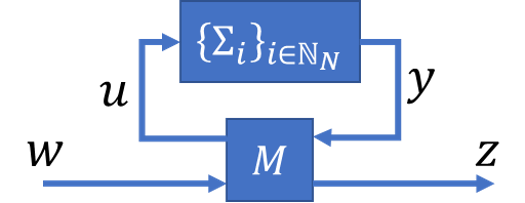}
	\caption{The networked system $\Sigma$.}
	\label{Fig:InterconnectedSystems}
\end{figure}

We make the following assumptions regarding the subsystems. Each subsystem $\Sigma_i, i\in\N_N$: 
(i) follows the dynamics
\begin{equation}
    \begin{aligned}
        x_i(t+1) = f_i(x_i(t),u_i(t)),\\
        y_i(t) = h_i(x_i(t),u_i(t)),
    \end{aligned}
\end{equation}
where $x_i(t)\in\R^{n_{xi}}, u_i(t)\in\R^{n_{ui}}$ and $y_i(t)\in\R^{n_{yi}}$; 
(ii) has a set $\mathcal{X}_i \subset \R^{n_{xi}}$ where for every $x_i^* \in \mathcal{X}_i$, there is a unique $u_i^* \in \R^{n_{ui}}$ satisfying $x_i^* = f_i(x_i^*,u_i^*)$, and both $u_i^*$ and $y_i^* \triangleq h_i(x_i^*,u_i^*)$ are implicit functions of $x_i^*$;  and 
(iii) is $X_i$-EID where $X_i\triangleq [X_i^{kl}]_{k,l\in\N_2}$.

The interconnection matrix $M$ is assumed to be structured as given in the interconnection relationship:
\begin{align}
	\label{Eq:NSC2Interconnection}\bm{u\\z} = M \bm{y\\w} \equiv 
	\bm{M_{uy} & M_{uw} \\ M_{zy} & M_{zw}}\bm{y\\w},
\end{align}
where $u \triangleq [u_i^\T]_{i\in\N_N}^\T \in \R^{n_u}$ and $y \triangleq [y_i^\T]_{i\in\N_N}^\T \in \R^{n_y}$ with $n_u \triangleq \sum_{i\in\N_N} n_{ui}$ and $n_y \triangleq \sum_{i\in\N_N} n_{yi}$. 

The networked system $\Sigma$ is assumed to have a set of equilibrium states $\mathcal{X} \subseteq \bar{\mathcal{X}} \in \R^{n_x}$ where $\bar{\mathcal{X}} \triangleq \mathcal{X}_1 \times \mathcal{X}_2 \times \cdots \times \mathcal{X}_N$ and $n_x \triangleq \sum_{i\in\N_N} n_{xi}$. The subsystem equilibrium points implies that, for each $x^* \triangleq [x_i^{*\T}]^\T_{i\in\N_N}\in\mathcal{X}$, there is a corresponding $u^* \triangleq [u_i^{*\T}]^\T$ and $y^* \triangleq [y_i^{*\T}]^\T$ that are implicit functions of $x^*$. Also, for each $x^*\in\mathcal{X}$, there is a unique $w^*\in\R^{n_w}$ that satisfies $u^* = M_{uy}y^* + M_{uw} w^*$, and $w^*$ and $z^* =  M_{zy}y^* + M_{zw}w^*$ are implicit functions of $x^*$.

\begin{figure*}[!b]
\centering
\hrulefill
\begin{equation}\label{Eq:Pr:NSC2Synthesis2}
\scriptsize 
\bm{
\textbf{X}_p^{11} & \0 & L_{uy} & L_{uw} \\
\0 & -\textbf{Y}^{22} & -\textbf{Y}^{22}M_{zy} & -\textbf{Y}^{22} M_{zw}\\ 
L_{uy}^\T & -M_{zy}^\T\textbf{Y}^{22} & -L_{uy}^\T \textbf{X}^{12}-\textbf{X}^{21}L_{uy}-\textbf{X}_p^{22} & -\textbf{X}^{21}L_{uw}+M_{zy}^\T\textbf{Y}^{21} \\
L_{uw}^\T & -M_{zw}^\T\textbf{Y}^{22} & -L_{uw}^\T \textbf{X}^{12}+\textbf{Y}^{12} M_{zy} &  M_{zw}^\T\textbf{Y}^{21} +  \textbf{Y}^{12}M_{zw} + \textbf{Y}^{11}
} \normalsize >0
\end{equation}
\begin{equation}\label{Eq:Pr:NSC2Synthesis2Alternative} 
\scriptsize 
\bm{
    -\textbf{Y}^{22} & -\textbf{Y}^{22}M_{zy} & -\textbf{Y}^{22}M_{zw} \\
    -M_{zy}^\T\textbf{Y}^{22} & -L_{uy}^\T \textbf{X}^{12} - \textbf{X}^{21}L_{uy} - \textbf{X}_p^{22} + \alpha^2\textbf{X}_p^{11}-\alpha(L_{uy}^\T + L_{uy}) & -\textbf{X}^{21}L_{uw} + M_{zy}^\T \textbf{Y}^{21} + \alpha^2\textbf{X}_p^{11}-\alpha(L_{uy}^\T + L_{uw}) \\
-M_{zw}^\T\textbf{Y}^{22} & -L_{uw}^\T \textbf{X}^{12} + \textbf{Y}^{12}M_{zy} + \alpha^2\textbf{X}_p^{11}-\alpha(L_{uw}^\T + L_{uy}) & M_{zw}^\T \textbf{Y}^{21} +  \textbf{Y}^{12}M_{zw} + \textbf{Y}^{11} + \alpha^2\textbf{X}_p^{11}-\alpha(L_{uw}^\T + L_{uw}) 
} \normalsize >0,  \alpha \in \R
\end{equation}
\end{figure*}

Following \cite{WelikalaP72023,Welikala2023Ax6}, we next provide an LMI problem for synthesizing the interconnection matrix $M$ \eqref{Eq:NSC2Interconnection} such that the networked system $\Sigma$ is $\textbf{Y}$-EID. For this purpose, we first require the following two mild assumptions.  


\begin{assumption}\label{As:NegativeDissipativity}
The desired $\textbf{Y}$-EID specification for the networked system $\Sigma$ is such that $\textbf{Y}^{22}<0$. 
\end{assumption}

\begin{remark}\label{Rm:As:NegativeDissipativity}
According to Rm. \ref{Rm:X-DissipativityVersions}, As. \ref{As:NegativeDissipativity} holds if we require the networked system $\Sigma$ to be: (1) L2G($\gamma$), or (2) IF-OFP($\nu,\rho$) with some $\rho>0$, i.e., $L_2$-stable or passive, respectively. Since it is always desirable to make the networked system $\Sigma$ either $L_2$-stable or passive, As. \ref{As:NegativeDissipativity} is mild.
\end{remark}

\begin{assumption}\label{As:PositiveDissipativity}
Each subsystem $\Sigma_i, i\in\N_N$ in the networked system $\Sigma$ is $X_i$-EID with $X_i^{11} > 0$ or $X_i^{11} < 0$. 
\end{assumption}

\begin{remark}\label{Rm:As:PositiveDissipativity}
According to Rm. \ref{Rm:X-DissipativityVersions}, As. \ref{As:PositiveDissipativity} holds if each subsystem $\Sigma_i,i\in\N_N$ is: (1) L2G($\gamma_i$) (i.e., $L_2$-stable), or (2) IF-OFP($\nu_i,\rho_i$) with $\nu_i>0$ (i.e., passive) or $\nu_i<0$ (i.e., non-passive). If such conditions are not met, local controllers can be used (e.g., via Co. \ref{Co:LTISystemXDisspativation}). Therefore, As.  \ref{As:PositiveDissipativity} is also mild. 
\end{remark}

\begin{proposition}\label{Pr:NSC2Synthesis}
Under Assumptions \ref{As:NegativeDissipativity} and \ref{As:PositiveDissipativity}, the discrete-time networked system $\Sigma$ is $\textbf{Y}$-EID if the interconnection matrix $M$ \eqref{Eq:NSC2Interconnection} is synthesized using the LMI problem:
\begin{equation}\label{Eq:Pr:NSC2Synthesis}
    \begin{aligned}
    \mbox{Find: }& L_{uy}, L_{uw}, M_{zy}, M_{zw}, \{p_i: i\in\N_N\}\\
    \mbox{Sub. to: }& p_i > 0, \forall i\in\N_N, \mbox{ and }\\
    & \begin{cases} 
    \eqref{Eq:Pr:NSC2Synthesis2} \ \ \ \mbox{if }      &X_i^{11}>0, \forall i\in\N_N,\\ 
    \eqref{Eq:Pr:NSC2Synthesis2Alternative} \ \ \ \mbox{else if}  &X_i^{11}<0, \forall i\in\N_N, 
    \end{cases}
    \end{aligned}
\end{equation}
where $\textbf{X}^{12} \triangleq \diag((X_i^{11})^{-1}X_i^{12}:i\in\N_N)$, $\textbf{X}^{21} \triangleq (\textbf{X}^{12})^\T$ with 
$M_{uy} \triangleq (\textbf{X}_p^{11})^{-1} L_{uy}$ and $M_{uw} \triangleq  (\textbf{X}_p^{11})^{-1} L_{uw}$.
\end{proposition}
\begin{proof}
See Appendix \ref{App:Pr:NSC2Synthesis}.
\end{proof}

\begin{remark}
The synthesis technique for the interconnection matrix $M$ given in Prop. \ref{Pr:NSC2Synthesis} can be used even when $M$ is partially known/fixed (as it will only reduce the number of variables in the LMI problem \eqref{Eq:Pr:NSC2Synthesis}). Further, it is independent of the equilibrium points of the considered networked system. Furthermore, it can be easily modified to handle scenarios where a subset of subsystems have $X_i^{11}>0$ while the others have $X_i^{11}<0$. 
\end{remark}

\section{Inventory Consensus Control Problem Formulation}\label{Sec:Consensus}

As mentioned in the Introduction, coordinating several parallel supply chains to maintain a consensus improves the overall performance and robustness of the supply chain network (SCN) \cite{Li2020}. Maintaining such a consensus can also be viewed as a synchronizing action across the SCN - which is particularly vital when individual supply chains are physically identical and supply the same product. Another scenario where consensus/synchrony is crucial is when different supply chains are dedicated to providing different products but at some fixed ratio. 
In this section, we formulate the inventory consensus control problem as a networked system control problem so that the dissipativity-based control and topology co-design technique in Prop. \ref{Pr:NSC2Synthesis} can be applied.

\begin{figure}[!t]
	\centering
	\includegraphics[width=\columnwidth]{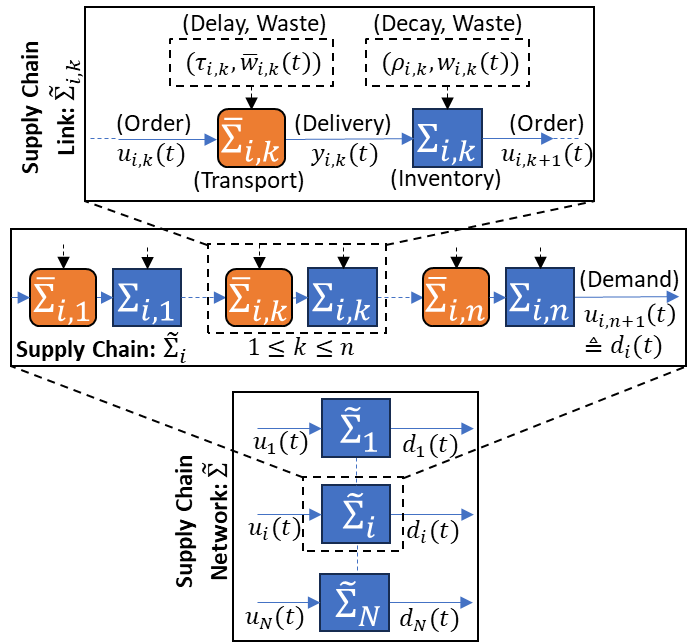}
	\caption{The composition of the supply chain network.}
	\label{Fig:SupplyChainConsensusProblem}
\end{figure}

\subsection{\textbf{Supply Chain Dynamics}} 

As shown in Fig. \ref{Fig:SupplyChainConsensusProblem}, we consider a SCN comprised of $N$ parallel supply chains $\{\tilde{\Sigma}_i:i\in\N_N\}$, where each supply chain $\tilde{\Sigma}_i$, $i\in\N_N$ is comprised of $n$ sequentially connected supply chain links $\{\tilde{\Sigma}_{i,k},k\in\N_n\}$. Each supply chain link $\tilde{\Sigma}_{i,k}, i\in\N_N, k\in\N_n$ includes a transportation process $\bar{\Sigma}_{i,k}$ and a subsequent inventory/storage process $\Sigma_{i,k}$.

At each supply chain link $\tilde{\Sigma}_{i,k}, i\in\N_N, k\in\N_n$, the inventory level, placed order magnitude, and received delivery amount at time step $t\in\N$ are denoted by $x_{i,k}(t)$, $u_{i,k}(t)$, and $y_{i,k}(t)$, respectively, and the transportation delay and inventory decay rate parameters are denoted by $\tau_{i,k}$ and $\rho_{i,k}$ respectively. Using these notations, for each supply chain link $\tilde{\Sigma}_{i,k}, i\in\N_N, k\in\N_n$, the transportation dynamics can be written as 
\begin{equation}\label{Eq:LinkTransportationDynamics}
    y_{i,k}(t) = u_{i,k}(t-\tau_{i,k}) - \bar{w}_{i,k}(t), 
\end{equation} 
where $\bar{w}_{i,k}(t)$ is the observed transportation loss (waste) at time step $t\in\N$, and the inventory dynamics can be written as 
\begin{equation}\label{Eq:LinkInventoryDynamics}
x_{i,k}(t+1) = (1-\rho_{i,k})x_{i,k}(t) + y_{i,k}(t) - u_{i,k+1}(t) - w_{i,k}(t),
\end{equation}
where $w_{i,k}(t)$ is the observed inventory loss (waste) at time step $t\in\N$. Note that, for $k=n$ (i.e., for the supply chain link $\tilde{\Sigma}_{i,n}$ at the edge of the supply chain $\tilde{\Sigma}_i$), we consider  
\begin{equation}
u_{i,n+1}(t) \triangleq d_{in}(t),
\end{equation}
where $d_{in}(t)$ represents the uncontrollable customer demand. Besides, the remaining order magnitudes $\{u_{i,k}(t): k\in\N_n\}$ are considered as the controllable inputs.

The transportation delays $\{\tau_{i,k}: i\in\N_N, k\in\N_N\}$ are assumed to be time-invariant and integer-valued (i.e., an integer multiple of the time step size used). Similarly, the inventory decay rates (also called the ``perish rates,'' and typically falls in the interval $[0,1]$) are assumed to be time-invariant. Note that these two assumptions are mild as any violations can be compensated by the respective loss (waste) terms introduced in the transportation \eqref{Eq:LinkTransportationDynamics} and inventory \eqref{Eq:LinkInventoryDynamics} dynamics. Note also that these loss terms can also compensate for any modeling imperfections in \eqref{Eq:LinkTransportationDynamics}-\eqref{Eq:LinkInventoryDynamics}.

Using the previous assumption that $\tau_{i,k}\in\N$, we next express the transportation dynamics \eqref{Eq:LinkTransportationDynamics} in an LTI system form. For this purpose, we introduce transportation states:
\begin{equation}\label{Eq:LinkTransportationStates}
\begin{aligned}
\bar{x}_{i,k}^{(1)}(t) \triangleq&\ u_{i,k}(t-\tau_{i,k}),\\
\bar{x}_{i,k}^{(2)}(t) \triangleq&\ u_{i,k}(t-\tau_{i,k}+1),\\
\vdots&\ \\
\bar{x}_{i,k}^{(\tau_{i,k})}(t) \triangleq&\ u_{i,k}(t-1).
\end{aligned}
\end{equation}
Then, by vectorizing these states $\bar{x}_{i,k}(t) \triangleq [\bar{x}_{i,k}^{(l)}(t)]_{l\in\N_{\tau_{i,k}}}^\T$, \eqref{Eq:LinkTransportationDynamics} can be expressed as a $\tau_{i,k}$-dimensional LTI system:
\begin{equation}\label{Eq:LinkTransportationDynamics2}
\begin{aligned}
\bar{x}_{i,k}(t+1) =&\ \bar{A}_{i,k} \bar{x}_{i,k}(t) + \bar{B}_{i,k}u_{i,k}(t),\\
y_{i,k}(t) =&\ \bar{C}_{i,k}\bar{x}_{i,k}(t) - \bar{w}_{i,k}(t),
\end{aligned}
\end{equation}
where 
\begin{equation*}
\begin{aligned}
\bar{A}_{i,k} \triangleq&\  
\bm{
0&1&0&\cdots&0\\
0&0&1&\cdots&0\\
\vdots & \vdots& \vdots& \ddots & \vdots\\
0&0&0&\cdots&1\\
0&0&0&\cdots&0}, \ \ 
\bar{B}_{i,k} \triangleq \bm{0\\0\\\vdots\\0\\1},\\
\bar{C}_{i,k} \triangleq&\ \bm{1&0&\cdots&0&0}. 
\end{aligned}
\end{equation*}

We are now ready to express the dynamics of an entire supply chain $\tilde{\Sigma}_i,i\in\N_N$ in a compact form by considering vectorized versions of \eqref{Eq:LinkInventoryDynamics} and \eqref{Eq:LinkTransportationDynamics2} (across all $k\in\N_n$). First, let us introduce vectorized transportation states, inventory levels, order magnitudes, delivery amounts, transportation losses, inventory losses, and customer demands associated with the supply chain $\tilde{\Sigma}_i, i\in\N_N$ as: 
$\bar{x}_i(t) \triangleq [\bar{x}_{i,k}^\T(t)]_{k\in\N_n}^\T$, 
$x_i(t) \triangleq [x_{i,k}(t)]_{k\in\N_n}^\T$, 
$u_i(t) \triangleq [u_{i,k}(t)]_{k\in\N_n}^\T$,
$y_i(t) \triangleq [y_{i,k}(t)]_{k\in\N_n}^\T$,
$\bar{w}_i(t) \triangleq [\bar{w}_{i,k}]_{k\in\N_n}^\T$, 
$w_i(t) \triangleq [w_{i,k}(t)]_{k\in\N_n}^\T$ and 
$d_{i}(t) \triangleq [0, 0, \cdots, 0, d_{in}(t)]^\T$,
respectively. Using these notations, for each supply chain $\tilde{\Sigma}_i, i\in\N_N$, the transportation dynamics can be written as 
\begin{equation}\label{Eq:ChainTransportationDynamics}
\begin{aligned}
\bar{x}_{i}(t+1) =&\ \bar{A}_{i} \bar{x}_i(t) + \bar{B}_{i} u_i(t),\\
y_i(t) =&\ \bar{C}_i \bar{x}_i(t) - \bar{w}_i(t).  
\end{aligned}
\end{equation}
where 
$\bar{A}_{i} \triangleq \diag(\bar{A}_{i,k}:k\in\N_n)$, 
$\bar{B}_{i} \triangleq \diag(\bar{B}_{i,k}:k\in\N_n)$ and 
$\bar{C}_{i} \triangleq \diag(\bar{C}_{i,k}:k\in\N_n)$,
and the inventory dynamics can be written as 
\begin{equation}\label{Eq:ChainInventoryDynamics}
x_i(t+1) = A_i x_i(t) + y_i(t) - B_i u_i(t) - d_i(t) - w_i(t),
\end{equation} 
where 
$A_i \triangleq \diag((1-\rho_{i,k}):k\in\N_n)$ and 
$$   
B_i \triangleq 
\bm{
    0&1&0&\cdots&0\\
    0&0&1&\cdots&0\\
    \vdots & \vdots& \vdots& \ddots & \vdots\\
    0&0&0&\cdots&1\\
    0&0&0&\cdots&0}.
$$
Finally, it is worth noting the dimensions: $\bar{x}_i(t) \in \R^{\tau_i}$ where $\tau_i\triangleq \sum_{k\in\N_n}\tau_{i,k}$ (i.e., the total delay in the supply chain $\tilde{\Sigma}_i$), and $u_i(t), y_i(t), \bar{w}_i(t), x_i(t), d_i(t), w_i(t) \in \R^n$.

\subsection{\textbf{Steady-State Control}} 
For generality, the disturbance vectors $\bar{w}_i(t)$, $w_i(t)$ and $d_i(t)$ in \eqref{Eq:ChainTransportationDynamics}-\eqref{Eq:ChainInventoryDynamics} are assumed to be random vectors with known fixed means $\overline{\bar{w}}_i$, $\overline{w}_i$ and $\overline{d}_i$, respectively. And thus, we can express $\bar{w}_i(t) = \overline{\bar{w}}_i + \tilde{\bar{w}}_i(t)$, $w_i(t) = \overline{w}_i + \tilde{w}_i(t)$ and $d_i(t) = \overline{d}_i + \tilde{d}_i(t)$, where $\tilde{\bar{w}}_i(t)$, $\tilde{w}_i(t)$ and $\tilde{d}_i(t)$ represent the corresponding zero mean disturbance vectors, respectively.

On the other hand, let us denote the steady-state values of the transportation state vector $\bar{x}_i(t)$ in \eqref{Eq:ChainTransportationDynamics} and the inventory state vector $x_i(t)$ in \eqref{Eq:ChainInventoryDynamics} by $\overline{\bar{x}}_i$ and  $\overline{x}_i$, respectively. It is worth mentioning that the steady-state inventory state vector $\overline{x}_i$ is a design variables that need to be prespecified. For example, in strategies like ``just-in-time,'' the steady-state inventory values are selected such that  $\overline{x}_i=\0$ \cite{Li2020}.

To support the said mean disturbance levels and steady-state inventory levels, each supply chain $\tilde{\Sigma}_i, i\in\N_N$ is assumed to select its control input (order magnitudes) as
\begin{equation}\label{Eq:SteadyStateControl0}
    u_i(t) = \overline{u}_i + \tilde{u}_i(t),
\end{equation} 
where $\overline{u}_i \triangleq [\overline{u}_{i,k}]_{k\in\N_n}^\top$ is the steady-state control input (i.e., the steady-state ordering magnitudes) and $\tilde{u}_i(t)$ is a virtual control law. The details of $\tilde{u}_i(t)$ will be provided later on. However, the following proposition provides the steady-state control input $\overline{u}_i$ and the steady-state transportation state $\overline{\bar{x}}_i$. 

\begin{proposition}\label{Pr:SteadyStateControl}
For each supply chain $\tilde{\Sigma}_i, i\in\N_N$ described by \eqref{Eq:ChainTransportationDynamics}-\eqref{Eq:ChainInventoryDynamics}, under fixed deterministic disturbances (i.e.,  $\bar{w}_i(t) = \overline{\bar{w}}_i$, $w_i(t) = \overline{w}_i$, and $d_i(t) = \overline{d}_i, \forall t\in\N$), there exists an equilibrium point such that 
\begin{equation}
\begin{aligned}
\bar{x}_i(t) =&\ \overline{\bar{x}}_i = \bar{D}_i\overline{u}_i, \\
x_i(t) =&\ \overline{x}_i \mbox{ (prescribed)},
\end{aligned}
\end{equation}
under steady-state control input $\overline{u}_i \triangleq [\overline{u}_{i,k}]_{k\in\N_n}^\top$ with 
\begin{equation}\label{Eq:SteadyStateControlLink}
\overline{u}_{i,k} = \sum_{j=k}^{n} (\rho_{i,k}\overline{x}_{i,k}+\overline{\bar{w}}_{i,k} + \overline{w}_{i,k}) + \overline{d}_i,\quad \forall k\in\N_n
\end{equation}
and $\tilde{u}_i(t) = \0$ in \eqref{Eq:SteadyStateControl0} with $\bar{D}_i \triangleq \diag(\mb{1}_{\tau_{i,k}}:k\in\N_n)$.
\end{proposition}
\begin{proof}
Under the given conditions, $u_i(t) = \overline{u}_i$ in \eqref{Eq:SteadyStateControl0}. Consequently, starting from \eqref{Eq:ChainTransportationDynamics}, the resulting steady-state transportation state $\overline{\bar{x}}_i$ can be obtained as
\begin{align}
\label{Eq:SteadyStateTransportationState}
\overline{\bar{x}}_{i} =&\ \bar{A}_{i} \overline{\bar{x}}_i + \bar{B}_{i} \overline{u}_i 
= (\I-\bar{A}_i)^{-1}\bar{B}_i\overline{u}_i = \bar{D}_i\overline{u}_i.
\end{align}
The above last step is due to $(\I-\bar{A}_{i,k})^{-1}$ being an upper-triangular matrix of ones, for any $k\in\N_n$. 

On the other hand, starting from \eqref{Eq:ChainInventoryDynamics}, the resulting steady-state inventory state $\overline{x}_i$ can be obtained as
\begin{align}
\overline{x}_i =&\ A_i \overline{x}_i + (\bar{C}_i \overline{\bar{x}}_i - \overline{\bar{w}}_i) - B_i \overline{u}_i - \overline{d}_i - \overline{w}_i\nonumber\\
=&\ (\I-A_i)^{-1}\left( (\bar{C}_i \overline{\bar{x}}_i - \overline{\bar{w}}_i) - B_i \overline{u}_i - \overline{w}_i - \overline{d}_i\right)\nonumber\\
=&\ (\I-A_i)^{-1}\left( (\bar{C}_i \bar{D}_i - B_i) \overline{u}_i - \overline{\bar{w}}_i - \overline{w}_i - \overline{d}_i\right)\nonumber\\
=&\ (\I-A_i)^{-1}\left( (\I - B_i) \overline{u}_i - \overline{\bar{w}}_i - \overline{w}_i - \overline{d}_i\right).    
\label{Eq:SteadyStateInventoryState}
\end{align}
The above last two steps are are due to \eqref{Eq:SteadyStateTransportationState} and $\bar{C}_i\bar{D}_i=\I$, respectively. 
Now, since $\overline{x}_i$ is prespecified, the required steady-state control $\overline{u}_i$ can be obtained using \eqref{Eq:SteadyStateInventoryState} as 
\begin{equation}\label{Eq:SteadyStateControl}
\overline{u}_i = (\I - B_i)^{-1}\left((\I-A_i)\overline{x}_i + \overline{\bar{w}}_i + \overline{w}_i + \overline{d}_i\right).
\end{equation}
Finally, using the facts that $(\I-B_i)^{-1}$ is an upper-triangular matrix of ones and $(\I-A_i)=\diag(\rho_{i,k}:k\in\N_n)$, it is easy to see that \eqref{Eq:SteadyStateControl} is equivalent to \eqref{Eq:SteadyStateControlLink}.
\end{proof}

The result \eqref{Eq:SteadyStateControlLink} established in Prop. \ref{Pr:SteadyStateControl} intuitively means that the order magnitude of each supply chain link has to compensate for the total mean waste of the downstream supply chain and the mean customer demand.

\subsection{\textbf{Error Dynamics}}

We now provide the error dynamics of each supply chain $\tilde{\Sigma}_i, i\in\N_N$ around its equilibrium point under steady-state control \eqref{Eq:SteadyStateControl0} with \eqref{Eq:SteadyStateControl}. Let us define the error states 
$\tilde{\bar{x}}_i(t) \triangleq \bar{x}_i(t) - \overline{\bar{x}}_i$ and 
$\tilde{x}_i(t) \triangleq x_i(t) - \overline{x}_i$, and \emph{total disturbance input} 
$r_i(t)\triangleq \tilde{\bar{w}}_i + \tilde{w}_i(t) + \tilde{d}_i(t)$. Then, using  \eqref{Eq:ChainTransportationDynamics} and \eqref{Eq:ChainInventoryDynamics} (together with \eqref{Eq:SteadyStateTransportationState} and \eqref{Eq:SteadyStateInventoryState}), we get
\begin{equation}\label{Eq:SupplyChainErrorDynamics0}
\underbrace{\bm{\tilde{x}_i(t+1)\\\tilde{\bar{x}}_{i}(t+1)}}_{\mathrm{x}_i(t+1)}
= \underbrace{\bm{A_i & \bar{C}_i\\\0 & \bar{A}_{i}}}_{\mathcal{A}_i}\underbrace{\bm{\tilde{x}_i(t)\\\tilde{\bar{x}}_i(t)}}_{\mathrm{x}_i(t)} 
+ \underbrace{\bm{-B_i\\\bar{B}_i}}_{\mathcal{B}_i}\tilde{u}_i(t) + \underbrace{\bm{-\I\\\0}}_{\mathcal{D}_i} r_i(t).
\end{equation}
To stabilize (and dissipativate, if necessary) \eqref{Eq:SupplyChainErrorDynamics0}, a local state feedback controller $\mathcal{L}_i$ can be included by defining: 
\begin{equation}\label{Eq:LocalStateFeedbackController}
\tilde{u}_i(t) \triangleq \mathcal{L}_i \mathrm{x}_i(t) + \tilde{\tilde{u}}_i(t),
\end{equation}
where the virtual control input $\tilde{\tilde{u}}_i(t)$ is yet to be designed based on other global control needs like achieving supply chain inventory consensus. 
For the use of such global controllers, at each supply chain $\tilde{\Sigma}_i, i\in\N_n$, we also include the inventory error $\tilde{x}_i(t)$ as a \emph{virtual output}: 
$$\mathrm{y}_i(t) \triangleq \tilde{x}_i(t) = \mathcal{C}_i \mathrm{x}_i(t),$$
where $\mathcal{C}_i \triangleq \bm{\I & \0}$. Consequently, the error dynamics of each supply chain $\tilde{\Sigma}_i, i\in\N_N$ takes the compact form
\begin{equation}\label{Eq:SupplyChainErrorDynamics}
\begin{aligned}
\mathrm{x}_i(t+1) =&\ (\mathcal{A}_i + \mathcal{B}_i \mathcal{L}_i)\mathrm{x}_i(t) + \eta_i(t),\\
\mathrm{y}_i(t) =&\ \mathcal{C}_i \mathrm{x}_i(t),
\end{aligned}
\end{equation}
where we define the \emph{virtual input} as
\begin{equation}\label{Eq:SupplyChainErrorDynamicsInput0}
\eta_i(t) \triangleq  \mathcal{B}_i \tilde{\tilde{u}}_i(t) + \mathcal{D}_i r_i(t).
\end{equation}
The remaining task is to design the virtual control input $\tilde{\tilde{u}}_i(t)$ (also called the global controller), which we will address in the sequel. Before moving on, it is worth noting the dimensions of some key  variables: $\mathrm{x}_i(t), \eta_i(t) \in\R^{n_i}$ where $n_i \triangleq n + \tau_i$, and $r_i(t), \tilde{u}_i(t), \tilde{\tilde{u}}_i(t), \mathrm{y}_i(t) \in \R^n$.

\subsection{\textbf{Inventory Consensus Control}} 
The goal of inventory consensus is to maintain some coordination among the supply chain inventories so that 
\begin{equation}\label{Eq:CoordinationGoal0}
x_i(t) - x_j(t) = \rho_{ij}, \quad \forall i,j \in\N_N, i \neq j.
\end{equation}
Clearly, $\rho_{ij} \triangleq \overline{x}_i - \overline{x}_j$, and since $\mathrm{y}_i(t) = x_i(t) - \overline{x}_i$, \eqref{Eq:CoordinationGoal0} can be restated as
\begin{equation}\label{Eq:CoordinationGoal}
    \mathrm{y}_i(t) - \mathrm{y}_j(t) = \0, \quad \forall i,j\in\N_N,i\neq j.
\end{equation} 
To achieve this goal, at each supply chain $\tilde{\Sigma}_i,\i\in\N_N$, we propose to use a distributed inventory consensus controller:
\begin{equation}\label{Eq:ConsensusControl}
    \tilde{\tilde{u}}_i(t) \triangleq \mathcal{L}_{ii} \mathrm{y}_i(t) + \sum_{j\in \N_N\backslash\{i\}} \mathcal{L}_{ij}(\mathrm{y}_i(t)-\mathrm{y}_j(t)).
\end{equation}
Typically (e.g., see \cite{Li2020}), each $\mathcal{L}_{ij}$ term in \eqref{Eq:ConsensusControl} is constrained to be scalars $1$ or $0$ depending on whether the information from supply chain $\Sigma_j$ can be obtained at $\Sigma_i$ or not, respectively. However, we propose to treat $\mathcal{L}_{ij}$ as a free coupling weight matrix. Nevertheless, in our framework, if $\mathcal{L}_{ij}=\0$, it still implies that information from $\Sigma_j$ is not necessary at $\Sigma_i$. Therefore, it is clear that: (1) the proposed consensus control law \eqref{Eq:ConsensusControl} is distributed, (2) the overall control solution is far more flexible, and thus, can lead to less conservative behaviors, and (3) via determining $\mathcal{L}_{ij}$ terms, we can synthesize the required information flow topology.

Under \eqref{Eq:ConsensusControl}, the closed-loop error dynamics of each supply chain $\tilde{\Sigma}_i, i\in\N_N$ can be written as 
\begin{equation}\label{Eq:SupplyChainErrorDynamicsClosedLoop}
\begin{aligned}
\mathrm{x}_i(t+1) =&\ (\mathcal{A}_i + \mathcal{B}_i \mathcal{L}_i)\mathrm{x}_i(t) + \eta_i(t),\\
\mathrm{y}_i(t) =&\ \mathcal{C}_i \mathrm{x}_i(t),\\
\eta_i(t) =&\ \mathcal{B}_i \sum_{j\in\N_N} \mathcal{K}_{ij} \mathrm{y}_j(t) + \mathcal{D}_i r_i(t),
\end{aligned}
\end{equation}
with $\mathcal{K}_{ii} \triangleq \sum_{j\in\N_N}\mathcal{L}_{ij}$ and $\mathcal{K}_{ij} \triangleq -\mathcal{L}_{ij}, \forall j\in\N_N\backslash\{i\}$. It is worth noting that $\mathcal{L} \triangleq[\mathcal{L}_{ij}]_{i,j\in\N_N}$ is fully determined by $\mathcal{K}\triangleq[\mathcal{K}_{ij}]_{i,j\in\N_N}$.

Before moving on, motivated by the said coordination goal \eqref{Eq:CoordinationGoal}, we also define a \emph{performance output} $z_i(t)\in\R^n$ for each supply chain $\tilde{\Sigma}_i, i\in\N_N$ as 
\begin{equation}\label{Eq:ConsensusPerformance}
z_i(t) \triangleq \frac{1}{N}\sum_{j\in\N_N}(\mathrm{y}_i(t)-\mathrm{y}_j(t)) \equiv \sum_{j\in\N_N} \mathcal{E}_{ij}\mathrm{y}_j(t)
\end{equation}
where $\mathcal{E}_{ij} \triangleq (\mb{1}_{\{i=j\}}-\frac{1}{N})\I_n, \forall i,j\in\N_N$. This performance objective is utilized in the sequel to design the optimally robust consensus controller $\mathcal{K}$ along with the required information flow topology (i.e., the communication topology).

\subsection{\textbf{Networked System View}} 
Using the aforementioned notions of \emph{virtual inputs} $\eta(t) \triangleq [\eta_i^\T(t)]_{i\in\N_N}^\T$, \emph{virtual outputs} $\mathrm{y}(t)\triangleq [\mathrm{y}_i^\T(t)]_{i\in\N_N}^\T$, \emph{total disturbance inputs} $r(t) \triangleq [r_i^\T(t)]_{i\in\N_N}^\T$ and \emph{performance outputs} $z(t) \triangleq [z_i^\T(t)]_{i\in\N_N}^\T$, the closed-loop error dynamics of the SCN (consisting of $N$ parallel supply chains) can be viewed as a networked system as shown in Fig. \ref{Fig:SupplyChainProblem}. Note that it is of the same form as the interconnected system considered in Fig. \ref{Fig:InterconnectedSystems}. In particular, here, the corresponding interconnection relationship takes the form: 
\begin{equation}\label{Eq:SupplyChainNetworkInterconnection}
	\bm{\eta(t)\\z(t)} = M \bm{\mathrm{y}(t)\\r(t)}\equiv \bm{M_{\eta \mathrm{y}} & M_{\eta r}\\M_{z\mathrm{y}} & M_{zr}}\bm{\mathrm{y}(t)\\r(t)},
\end{equation}
where $M_{\eta \mathrm{y}} = \mathcal{B}\mathcal{K}$ with $\mathcal{B} \triangleq \diag(\mathcal{B}_i:i\in\N_N)$, 
$M_{\eta r} = \mathcal{D}$ with $\mathcal{D} \triangleq \diag(\mathcal{D}_i:i\in\N_N)$, 
$M_{z\mathrm{y}} = \mathcal{E} \triangleq [\mathcal{E}_{ij}]_{i,j\in\N_N}$ and $M_{zr}=\0$. 
It is worth noting the dimensions: $\eta(t) \in \R^{\bar{n}}$ where $\bar{n} = \sum_{i\in\N_N} n_i$, and $z(t), \mathrm{y}(t), r(t) \in \R^{nN}$. 

\begin{figure}[!t]
	\centering
	\includegraphics[width=1.5in]{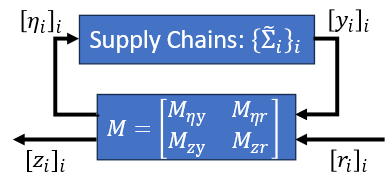}
	\caption{Closed-loop error dynamics of the SCN as a networked system.}
	\label{Fig:SupplyChainProblem}
\end{figure}

Through this networked system view of the closed-loop SCN, our goal is to synthesize the interconnection matrix block $M_{\eta \mathrm{y}}$ to obtain $\mathcal{L}_{ij}$ matrices as required in the used inventory consensus controller \eqref{Eq:ConsensusControl}. Since $\mathcal{L}_{ij}$ matrices are fully determined by the corresponding $\mathcal{K}_{ij}$ matrices, and each $\mathcal{K}_{ij}\in\R^{n\times n}$, the format of the matrix block $M_{\eta \mathrm{y}}$ is such that 
$
M_{\eta \mathrm{y}} \triangleq 
[M_{\eta \mathrm{y}}^{ij}]_{i,j\in\N_N}
$
where 
\begin{equation}\label{Eq:InterconnectionMatrixStructure}
M_{\eta \mathrm{y}}^{ij} 
= \mathcal{B}_i \mathcal{K}_{ij} \in \R^{n_i \times n},
\quad \forall i,j\in \N_N.  
\end{equation}
Consequently, we require each column of $M_{\eta \mathrm{y}}^{ij}$ to be in the column space of $\mathcal{B}_i$. This constraint, together with the facts that $\mathcal{B}_i \in \R^{n_i \times n}$ and $\text{rank }\mathcal{B}_i = n$, ensures that we can uniquely determine each $\mathcal{K}_{ij}$ using the corresponding $M_{\eta \mathrm{y}}^{ij}$.

Overall, the resulting inventory consensus controller \eqref{Eq:ConsensusControl} (through synthesizing $M_{\eta \mathrm{y}}$) should not only stabilize the closed-loop SCN but also attenuate the effects of disturbances $r(t)$ on the consensus performance $z(t)$. As the considered disturbances in our approach include both demand and intermediate uncertainties (e.g., waste), such an attenuation will minimize both the bullwhip effect and the ripple effect in the closed-loop SCN \eqref{Eq:SupplyChainErrorDynamicsClosedLoop}-\eqref{Eq:SupplyChainNetworkInterconnection}.

Consequently, when synthesizing $M_{\eta \mathrm{y}}$ (i.e., the consensus controllers), we enforce a prespecified bound $\bar{\gamma}$ on the squared $L_2$-gain $\gamma^2$ from disturbances $r(t)$ to consensus performance $z(t)$ in the closed-loop SCN. In other words, we enforce the constraint:
\begin{equation}\label{Eq:BullwhipEffect}
    \frac{\Vert z(\cdot) \Vert_{\mathcal{L}_2}^2}{\Vert r(\cdot) \Vert_{\mathcal{L}_2}^2} \leq \gamma^2 \mbox{ with } \gamma^2\leq \bar{\gamma}.
\end{equation}
According to Rm. \ref{Rm:X-DissipativityVersions}, this specification \eqref{Eq:BullwhipEffect} can be met by enforcing $\textbf{Y}\equiv \scriptsize \bm{\gamma^2 & \0\\ \0 & -\I}$-EID (from $r(t)$ to $z(t)$) along with the constraint $\gamma^2\leq \bar{\gamma}$ in the closed-loop SCN \eqref{Eq:SupplyChainErrorDynamicsClosedLoop}-\eqref{Eq:SupplyChainNetworkInterconnection}.

\section{Dissipativity-Based Control and Topology Co-Design Solution}
\label{Sec:ConsensusSolution}

This section outlines our dissipativity-based consensus control and communication topology co-design framework. As detailed in Sec. \ref{Sec:Consensus}, this co-design task can be achieved by simply designing the interconnection matrix block $M_{\eta \mathrm{y}}$ \eqref{Eq:SupplyChainNetworkInterconnection}, while ensuring it meets the necessary structural constraints \eqref{Eq:InterconnectionMatrixStructure} and the desired robustness constraints  \eqref{Eq:BullwhipEffect}. To this end, our first step is to use Co. \ref{Co:LTISystemXDisspativation} to enforce dissipativity (in particular, EID) at each supply chain error dynamics \eqref{Eq:SupplyChainErrorDynamics} via designing respective local controllers \eqref{Eq:LocalStateFeedbackController}. Subsequently, we use Prop. \ref{Pr:NSC2Synthesis} to design the interconnection matrix block $M_{\eta \mathrm{y}}$ \eqref{Eq:SupplyChainNetworkInterconnection} for the closed-loop SCN \eqref{Eq:SupplyChainErrorDynamicsClosedLoop}-\eqref{Eq:SupplyChainNetworkInterconnection}.

\subsection{Local Controller Design}

The following theorem formulates the local controller design problem as an LMI problem. The goal of each local controller is to enforce a certain EID property, which will then be exploited to design the global controller.

\begin{theorem}\label{Th:LocalControlDesign}
At each supply chain $\tilde{\Sigma}_i, i\in\N_n$, the error dynamics \eqref{Eq:SupplyChainErrorDynamics} are $X_i$-EID (from $\eta_i(t)$ to $\mathrm{y}_i(t)$) when the local controller $\mathcal{L}_i$ \eqref{Eq:LocalStateFeedbackController} is designed via the LMI problem:
\begin{equation}\label{Eq:Th:LocalControlDesign}
\begin{aligned}
&\mbox{Find:}&&\      \mathcal{K}_i, P_i, X_i\\
&\mbox{Sub. to:}&&\   P_i > 0, \mbox{ and }
\end{aligned}    
\end{equation}
\begin{equation}\label{Eq:Th:LocalControlDesign0}
\scriptsize
\bm{-(X_i^{22})^{-1} & \0 & \mathcal{C}_i P_i & \0 \\
\0 & P_i & \mathcal{A}_iP_i+\mathcal{B}_i\mathcal{K}_i & \I \\
P_i\mathcal{C}_i^\T & P_i\mathcal{A}_i^\T + \mathcal{K}_i^\T \mathcal{B}_i^\T & P_i  & P_i\mathcal{C}_i^\top X_i^{21}\\
\0 & \I & X_i^{12} \mathcal{C}_i P_i & X_i^{11}} 
\normalsize  \geq 0,
\end{equation}
where $\mathcal{L}_i = \mathcal{K}_iP_i^{-1}$.
\end{theorem}
\begin{proof}
The proof follows from considering the supply chain error dynamics \eqref{Eq:SupplyChainErrorDynamics} and applying the LMI-based controller design technique established in Co. \ref{Co:LTISystemXDisspativation}. 
\end{proof}

We next make several remarks regarding the LMI problem \eqref{Eq:Th:LocalControlDesign} in Th. \ref{Th:LocalControlDesign}, particularly pertaining to its implementation.

\begin{remark}
Recall that $n$ is the number of inventories in a supply chain and $n_i = n + \tau_i$ where $\tau_i$ is the total delay of transportation links in supply chain $\tilde{\Sigma}_i, i\in\N_N$. The dimensions of different variables/parameters appearing in the LMI problem \eqref{Eq:Th:LocalControlDesign} in Th. \ref{Th:LocalControlDesign} are as follows:  
$\mathcal{K}_i \in \R^{n \times n_i}$, $P_i \in \R^{n_i\times n_i}$,  
$\mathcal{C}_i \in \R^{n \times n_i}$,
$\mathcal{B}_i \in \R^{n_i \times n}$, 
$\mathcal{L}_i \in \R^{n \times n_i}$, and 
$X_i \triangleq [X_i^{kl}]_{k,l\in\N_2} \in \R^{(n_i+n) \times (n_i+n)}$ where 
$X_i^{11} \in \R^{n_i \times n_i}$, 
$X_i^{12} \in \R^{n_i\times n}$, 
$X_i^{21} \in \R^{n \times n_i}$, and 
$X_i^{22} \in \R^{n \times n}$.  
\end{remark}

\begin{remark}
    \label{Rm:LocalDissipativity}
When implementing Th. \ref{Th:LocalControlDesign}, we constrain the desired $X_i$-EID property to resemble the IF-OFP property mentioned in Rm. \ref{Rm:X-DissipativityVersions}. Therefore, $X_i = \scriptsize \bm{-\nu_i\I & 0.5\I \\ 0.5\I & -\rho_i\I}$, and the term $-(X_i^{22})^{-1}$ in \eqref{Eq:Th:LocalControlDesign0} becomes $-(X_i^{22})^{-1} = \frac{1}{\rho_i}\I = \tilde{\rho}_i \I$ where we define a linearizing change of variables: $\tilde{\rho}_i \triangleq \frac{1}{\rho_i}$ to preserve the LMI form of \eqref{Eq:Th:LocalControlDesign}. Consequently, instead of searching for a matrix $X_i$, we only have to search for scalar passivity indices $\nu_i$ and $\tilde{\rho}_i$ in the LMI problem \eqref{Eq:Th:LocalControlDesign}.
\end{remark}

\begin{remark}\label{Rm:LocalDissipativity2}
To enforce an IF-OFP property at each supply chain $\tilde{\Sigma}_i, i\in\N_N$ as mentioned in Rm. \ref{Rm:LocalDissipativity} requires us to extend its number of outputs from $n$ to $n_i$ so that the number of inputs and outputs are equal. This is simply achieved by replacing $\mathcal{C}_i \in \R^{n \times n_i}$ with $\I \in\R^{n_i \times n_i}$ and doing necessary dimension changes to $X_i^{22}, X_i^{12}$ and $X_i^{21}$ blocks in \eqref{Eq:Th:LocalControlDesign}. The intuition behind such a modification is to use the full error state (as opposed to a partial error state $\mathrm{y}_i(t) = \mathcal{C}_i \mathrm{x}_i(t)$) as the output. Including such additional outputs in the local controller design ensures a stricter dissipativity property at each supply chain and leads to better inventory tracking control results. It is also worth noting that this inclusion does not impact the networked system view of the SCN presented before or the global control design process discussed in the sequel. Overall, this modification motivates the systematic handling of inventory tracking and consensus control goals through local and global controllers, respectively. 
\end{remark}


\subsection{Global Control and Topology Co-Design}

The following theorem formulates the global control and topology co-design problem as an LMI problem. Note that, in this co-design, we optimize a joint control and topology objective. The control objective is selected as the parameter $\tilde{\gamma} \triangleq \gamma^2$ in the enforced robustness constraint \eqref{Eq:BullwhipEffect} to minimize the bullwhip and ripple effects in the SCN. 
The topology objective is selected analogous to 
$
\sum_{i,j\in\N_N, k,l\in\N_n} C_{ij}^{kl} \vert \mathcal{K}_{ij}^{kl} \vert
$ 
where 
$C \triangleq [[C_{ij}^{kl}]_{k,l\in\N_n}]_{i,j\in\N_N} \in \R^{nN\times nN}_{\geq 0}$ is a prespecified cost coefficient matrix containing communication cost values between different inventory pairs in the SCN. We use 1-norms (i.e., absolute values) of the design variables in the topology objective to induce sparsity \cite{Sun2020} in the resulting topology. Note also that we use a prespecified scaling cost coefficient $c_0 \in \R_{\geq 0}$ to scale up/down the control objective with respect to the topology objective.

\begin{theorem}\label{Th:GlobalControlDesign}
Under the local controllers designed in Th. \ref{Th:LocalControlDesign}, the closed-loop SCN \eqref{Eq:SupplyChainErrorDynamicsClosedLoop}-\eqref{Eq:SupplyChainNetworkInterconnection} (shown in Fig. \ref{Fig:SupplyChainProblem}) satisfies the robustness constraint \eqref{Eq:BullwhipEffect} when the distributed consensus controller $\mathcal{K}$ \eqref{Eq:SupplyChainErrorDynamicsClosedLoop} is designed via the LMI problem:
\begin{equation}\label{Eq:Th:GlobalControlDesign}
\begin{aligned}
\min_{\bar{\mathcal{K}}, \tilde{\gamma}, \{p_i:i\in\N_N\}}&\ 
\sum_{i,j\in\N_N, k,l\in\N_n} C_{ij}^{kl} \vert \bar{\mathcal{K}}_{ij}^{kl} \vert
+ c_0 \tilde{\gamma}\\
\mbox{Sub. to:}&\ 
p_i > 0, \forall i\in\N_N,\ \ 
0 \leq \tilde{\gamma} \leq \bar{\gamma}, \ \ \mbox{and} 
\end{aligned}
\end{equation}
\begin{equation}
\label{Eq:Th:GlobalControlDesign0}
\scriptsize 
\bm{
\textbf{X}_p^{11} & \0 & L_{\eta \mathrm{y}} & \textbf{X}_p^{11}\mathcal{D} \\
\0 & \I & \mathcal{E} & \0\\ 
L_{\eta \mathrm{y}}^\T & \mathcal{E}^\T & -L_{\eta \mathrm{y}}^\T \textbf{X}^{12}-\textbf{X}^{21}L_{\eta \mathrm{y}}-\textbf{X}_p^{22} & -\textbf{X}^{21}\textbf{X}_p^{11}\mathcal{D}\\
\mathcal{D}^\T \textbf{X}_p^{11} & \0 & -\mathcal{D}^\T \textbf{X}_p^{11} \textbf{X}^{12}&  \tilde{\gamma}\I
} \normalsize > 0,
\end{equation}
where $\bar{\mathcal{K}} \triangleq [\bar{\mathcal{K}}_{ij}]_{i,j\in\N_N}$ with each $\bar{\mathcal{K}}_{ij} \triangleq  [\bar{\mathcal{K}}_{ij}^{kl}]_{k,l\in\N_n}$,  
$\textbf{X}_p^{11} \triangleq \diag([p_iX_i^{11}]_{i\in\N_N})$,  
$\textbf{X}_p^{22} \triangleq \diag([p_iX_i^{22}]_{i\in\N_N})$,
$\textbf{X}^{12} \triangleq \diag([(X_i^{11})^{-1}X_i^{12}]_{i\in\N_N})$,
$\textbf{X}^{21} \triangleq (\textbf{X}^{12})^\T$, 
$L_{\eta \mathrm{y}} \triangleq [L_{\eta \mathrm{y}}^{ij}]_{i,j\in\N_N}$ with each 
$L_{\eta \mathrm{y}}^{ij} \triangleq X_i^{11} \mathcal{B}_i\bar{\mathcal{K}}_{ij}$, and 
\begin{equation}\label{Eq:Th:GlobalControlDesign1}
\mathcal{K} \triangleq [\mathcal{K}_{ij}]_{i,j\in\N_N} 
\mbox{ with each } \mathcal{K}_{ij} \triangleq \frac{1}{p_i}\bar{\mathcal{K}}_{ij}.    
\end{equation}
\end{theorem}
\begin{proof}
The proof starts with the interconnection topology synthesis result in Prop. \ref{Pr:NSC2Synthesis} and applying: (i) the subsystem dissipativity properties ensured in Th. \ref{Th:LocalControlDesign}, (ii) the desired closed-loop networked system dissipativity properties \eqref{Eq:BullwhipEffect}, and (iii) the SCN's interconnection matrices  \eqref{Eq:SupplyChainNetworkInterconnection}.  
The resulting main LMI condition takes the form:
\begin{equation}\label{Eq:Th:GlobalControlDesignStep1}
\scriptsize 
\bm{
\textbf{X}_p^{11} & \0 & L_{\eta y} & L_{\eta r} \\
\0 & \I & \E & \0 \\ 
L_{\eta y}^\T & \E^\T & -L_{\eta y}^\T \textbf{X}^{12}-\textbf{X}^{21}L_{\eta y}-\textbf{X}_p^{22} & -\textbf{X}^{21}L_{\eta r} \\
L_{\eta r}^\T & \0 & -L_{\eta r}^\T \textbf{X}^{12} &  \tilde{\gamma}\I
} \normalsize >0.
\end{equation}

Next, using the definition of $M_{uw}$ term given in Prop. \ref{Pr:NSC2Synthesis} and the relationship $M_{\eta r} = \mathcal{D}$ identified in \eqref{Eq:SupplyChainNetworkInterconnection}, we can simplify the $L_{\eta r}$ terms appearing in \eqref{Eq:Th:GlobalControlDesignStep1} as 
$L_{\eta r} = \textbf{X}_p^{11}M_{\eta r} = \textbf{X}_p^{11}\mathcal{D}$. 
Consequently, the LMI condition \eqref{Eq:Th:GlobalControlDesignStep1} can be simplified to obtain the LMI condition \eqref{Eq:Th:GlobalControlDesign0}.

Note that, under the definition of $M_{uy}$ given in Prop. \ref{Pr:NSC2Synthesis} and the relationship $M_{\eta y} = \mathcal{B}\mathcal{K}$ 
identified in \eqref{Eq:SupplyChainNetworkInterconnection}, we can interpret the $L_{\eta y}$ terms appearing in \eqref{Eq:Th:GlobalControlDesign0} as 
$L_{\eta y} = \textbf{X}_p^{11}M_{\eta y} = \textbf{X}_p^{11}\mathcal{B}\mathcal{K}$. This, under the proposed change of variables \eqref{Eq:Th:GlobalControlDesign1}, implies that
\begin{equation}\label{Eq:Th:GlobalControlDesignStep2}
    L_{\eta y}^{ij} = p_i X_i^{11}\mathcal{B}_i \mathcal{K}_{ij} =  X_i^{11}\mathcal{B}_i \bar{\mathcal{K}}_{ij}, \quad \forall i,j\in\N_N.
\end{equation}
Hence it is clear that all matrix blocks in \eqref{Eq:Th:GlobalControlDesign0} can be written as linear expressions of the decision variables $\bar{\mathcal{K}}, \tilde{\gamma}, \{p_i:i\in\N_N\}$ considered in the formulated co-design problem \eqref{Eq:Th:GlobalControlDesign}, thus preserving its LMI (convex optimization) problem form. 

Finally, we point out that the proposed change of variables \eqref{Eq:Th:GlobalControlDesign1} enables: (i) direct penalization of the used communication links, and (ii) direct evaluation of the required consensus controller gains $\mathcal{K} = [\mathcal{K}_{ij}]_{i,j\in\N_N}$ \eqref{Eq:SupplyChainErrorDynamicsClosedLoop}. 
As the motivation behind the proposed joint control and topology objective function used in \eqref{Eq:Th:GlobalControlDesign0} has been made clear in the preceding discussion to Th. \ref{Th:GlobalControlDesign}, this completes the proof. 
\end{proof}

\begin{figure*}[!b]
\hrulefill
\begin{align}\label{Eq:LocalControllerDesignConditionStep2}
\centering
\scriptsize
\bm{
p_iX_i^{11} & \0 & L_{\eta \mathrm{y}}^{ii} & p_iX_i^{11}\mathcal{D}_i \\
\0 & \I & \mathcal{E}_{ii} & \0\\ 
(L_{\eta \mathrm{y}}^{ii})^\T & \mathcal{E}_{ii}^\T &  -(L_{\eta \mathrm{y}}^{ii})^\T (X_i^{11})^{-1}X_i^{12}-X_i^{21}(X_i^{11})^{-1}L_{\eta \mathrm{y}}^{ii}-p_iX_i^{22} & -p_iX_i^{21}\mathcal{D}_i\\
p_i\mathcal{D}_{i}^\T X_i^{11} & \0 & -p_i\mathcal{D}_i^\T X_i^{12} &  \tilde{\gamma}_i\I
} > 0, \quad \forall i\in\N_N.
\end{align}
\end{figure*}

\subsection{Improved Local Controller Design}

The feasibility and the effectiveness of the formulated global co-design problem \eqref{Eq:Th:GlobalControlDesign} can be affected by the individual supply chain dissipativity properties (i.e., $X_i, i\in\N_N$) achieved by the local control design problems \eqref{Eq:Th:LocalControlDesign}. To force the resulting dissipativity properties from local control designs to pave the way to a feasible and effective global co-design, a special additional condition can be identified and included in the local control design. This idea is explained in the following two remarks, and the developed improved local control design is summarized in the proceeding Th. \ref{Th:ImprovedLocalControlDesign}.

\begin{remark}\label{Rm:ImprovedLocalDesign}
To identify a set of local necessary conditions for the feasibility of the global co-design problem \eqref{Eq:Th:GlobalControlDesign}, first, let us denote its main LMI constraint \eqref{Eq:Th:GlobalControlDesign0} as $\Psi > 0$. Based on the block structure of $\Psi$, note that $\Psi = [\Psi_{kl}]_{k,l\in\N_4}$ where each block $\Psi_{kl}=[\Psi_{kl}^{ij}]_{i,j\in \N_N}$. Using \cite[Lm. 1]{WelikalaP32022}, it can be shown that 
$$\Psi > 0 \iff \mbox{BEW}(\Psi) \triangleq [[\Psi_{kl}^{ij}]_{k,l\in\N_4}]_{i,j\in\N_N} > 0.$$ 
Now, considering the diagonal blocks of $\mbox{BEW}(\Psi)$, note that    
\begin{align*}\nonumber
    \mbox{BEW}(\Psi)>0 &\implies\ [\Psi_{kl}^{ii}]_{k,l\in\N_4} > 0, \ \forall i\in\N_N \\ \nonumber
    &\implies\ [\bar{\Psi}_{kl}^{ii}]_{k,l\in\N_4} > 0, \ \forall i\in\N_N 
    \ \equiv\ \eqref{Eq:LocalControllerDesignConditionStep2},
\end{align*}
where $\bar{\Psi}_{kl}^{ii}$ (as can be seen in  \eqref{Eq:LocalControllerDesignConditionStep2}) is identical to $\Psi_{kl}^{ii}$ except for a $\tilde{\gamma}_i$ term such that $\tilde{\gamma}_i > \tilde{\gamma}$ replacing the original $\tilde{\gamma}$ term in $\Psi_{kl}^{ii}, i\in\N_N$.
Consequently, \eqref{Eq:LocalControllerDesignConditionStep2} is a set of local necessary conditions for \eqref{Eq:Th:GlobalControlDesign0}, which now can be included in respective local control design problems \eqref{Eq:Th:LocalControlDesign}. It is worth noting that, in such improved local control design problems, the introduced variable $\tilde{\gamma}_i$ in \eqref{Eq:LocalControllerDesignConditionStep2} needs to be treated as an LMI variable to be determined, and the global design variables $p_i$ and $L_{\eta\mathrm{y}}^{ii}$ in \eqref{Eq:LocalControllerDesignConditionStep2} needs to be treated as prespecified design parameters. 
\end{remark}

At each supply chain $\tilde{\Sigma}_i, i\in\N_N$, if the specific $X_i$-EID property where $X_i = \scriptsize \bm{-\nu_i\I & 0.5\I \\ 0.5\I & -\rho_i\I}$ is adopted (as considered in Rm. \ref{Rm:LocalDissipativity}), as detailed in the following remark, a specialized set of local necessary conditions can be derived from \eqref{Eq:LocalControllerDesignConditionStep2}.

\begin{remark} \label{Rm:ImprovedLocalDesign2}
By applying the specific $X_i$-EID property where $X_i = \scriptsize \bm{-\nu_i\I & 0.5\I \\ 0.5\I & -\rho_i\I}$, we can further simplify the obtained local necessary conditions \eqref{Eq:LocalControllerDesignConditionStep2} into the form: 
\begin{equation}
\begin{aligned}
\scriptsize
\bm{
-p_i\nu_i\I & \0 & L_{\eta \mathrm{y}}^{ii} & -p_i\nu_i\mathcal{D}_i \\
\0 & \I_n & \mathcal{E}_{ii} & \0\\ 
(L_{\eta \mathrm{y}}^{ii})^\T & \mathcal{E}_{ii}^\T &  
\frac{1}{2\nu_i}(\bar{L}_{\eta \mathrm{y}}^{ii})^\T + \frac{1}{2\nu_i}\bar{L}_{\eta \mathrm{y}}^{ii} - p_i\rho_i\I_n & \frac{1}{2}p_i\I_n\\
-p_i\nu_i\mathcal{D}_{i}^\T & \0 & \frac{1}{2}p_i\I_n&  \tilde{\gamma}_i\I_n
} \normalsize 
> 0, \\
\forall i\in\N_N,
\end{aligned}
\label{Eq:LocalControllerDesignConditionStep3}
\end{equation}
where $\bar{L}_{\eta \mathrm{y}}^{ii} \triangleq \bm{\I_n & \0}L_{\eta \mathrm{y}}^{ii} \in \R^{n \times n}$. Next, using the facts that: 
(i) $\mathcal{D}_i = \scriptsize \bm{-\I_{n} \\ \0 }$, 
(ii) $\mathcal{E}_{ii} = (1-\frac{1}{N})\I_n$,
and  
(iii) the $n$\tsup{th} rows of $L_{\eta \mathrm{y}}^{ii}$ and $\bar{L}_{\eta \mathrm{y}}^{ii}$ are only of zeros (see \eqref{Eq:Th:GlobalControlDesignStep2}), a corresponding set of local necessary conditions for the above conditions \eqref{Eq:LocalControllerDesignConditionStep3}  can be identified as 
\begin{equation}\label{Eq:LocalControllerDesignConditionStep4}
\scriptsize
\bm{
-p_i\nu_i & 0 & 0 & p_i\nu_i \\
0 & 1 & (1-\frac{1}{N}) & 0\\ 
0 & (1-\frac{1}{N}) &  -p_i\rho_i & \frac{1}{2}p_i\\
p_i\nu_i & 0 & \frac{1}{2}p_i &  \tilde{\gamma}_i
}
\normalsize
> 0, \quad \forall \in \N_N.
\end{equation}
Finally, these specialized local necessary conditions \eqref{Eq:LocalControllerDesignConditionStep4} can be included in respective local control design problems \eqref{Eq:Th:LocalControlDesign}. Note that, in these improved local control design problems, unlike those in Rm, \ref{Rm:ImprovedLocalDesign}, only the global design variables $p_i$ in  \eqref{Eq:LocalControllerDesignConditionStep4} need to be treated as prespecified design parameters.
\end{remark}

We conclude this section by providing the improved local control design problem for the specific case considered in Remarks \ref{Rm:LocalDissipativity}, \ref{Rm:LocalDissipativity2} and \ref{Rm:ImprovedLocalDesign2}, as a theorem. The proof is omitted here, as it is evident from the said remarks.

\begin{theorem}\label{Th:ImprovedLocalControlDesign}
At each supply chain $\tilde{\Sigma}_i, i\in\N_n$, the error dynamics \eqref{Eq:SupplyChainErrorDynamics} are $X_i$-EID (from $\eta_i(t)$ to $\mathrm{x}_i(t)$) with $X_i = \scriptsize \bm{-\nu_i\I & 0.5\I \\ 0.5\I & -\rho_i\I}$ when the local controller $\mathcal{L}_i$ \eqref{Eq:LocalStateFeedbackController} is designed via the LMI problem:
\begin{equation}\label{Eq:Th:ImprovedLocalControlDesign}
\begin{aligned}
\mbox{Find:}\ &\mathcal{K}_i, P_i, \nu_i, \tilde{\rho}_i, \tilde{\gamma}_i\\
\mbox{Sub. to:}\ &P_i > 0, \mbox{ and }\\
&\scriptsize
\bm{-\tilde{\rho}_i\I & \0 & P_i & \0 \\
\0 & P_i & \mathcal{A}_iP_i+\mathcal{B}_i\mathcal{K}_i & \I \\
P_i & P_i\mathcal{A}_i^\T + \mathcal{K}_i^\T \mathcal{B}_i^\T & P_i  & 0.5P_i\\
\0 & \I & 0.5 P_i & \nu_i \I} 
\normalsize  \geq 0,\\
&\scriptsize
\bm{
-p_i\nu_i & 0 & 0 & p_i\nu_i \\
0 & 1 & (1-\frac{1}{N}) & 0\\ 
0 & (1-\frac{1}{N}) &  -p_i\rho_i & \frac{1}{2}p_i\\
p_i\nu_i & 0 & \frac{1}{2}p_i &  \tilde{\gamma}_i
}
\normalsize
\geq 0,
\end{aligned}    
\end{equation}
where $\mathcal{L}_i = \mathcal{K}_iP_i^{-1}$, $\rho_i = \frac{1}{\tilde{\rho}_i}$, and $p_i$ is some prespecified design parameter such that $p_i > 0$.
\end{theorem}

\section{Simulation Results} \label{Sec:SimulationResults}

In this section, we present several simulation results to demonstrate the effectiveness of the proposed SCN control strategy. To conduct the necessary SCN simulations, we specifically developed a highly customizable MATLAB-based SCN simulator (publicly available at Github: \href{https://github.com/shiran27/SupplyChainSimulator}{https://github.com/shiran27/SupplyChainSimulator}).

\subsection{Simulation Setup}

In the conducted experiments, as shown in Fig. \ref{Fig:InitialState}, we used a SCN with $N=3$ supply chains where each supply chain has $n=4$ inventories. 
When simulating this SCN, each discrete-time step was considered as an hour in real-time, and the total simulation length was chosen to be $T = 720$ times steps (i.e., $30$ days).  
The perish rate and the desired inventory level of each inventory $\Sigma_{i,k}$ were selected as $\rho_{i,k} = 0.1$ and $\overline{x}_{i,k} = 500$, respectively, for any $i\in\N_N, k\in\N_n$. 
The transportation delay of each transportation link $\bar{\Sigma}_{i,k}$ was selected randomly as $\tau_{i,k} = 1+\text{randi}(1,4)$ for any $i\in\N_N, k\in\N_n$ (the used delay values are clear from Fig. \ref{Fig:InitialState}). 
The initial states of the inventories and transportation links were selected randomly as $x_{i,k}(1) = \text{randi}(100,900)$ and $\bar{x}_{i,k}(1) = \text{randi}(100,900)$, respectively, for any $i\in\N_N, k\in\N_n$ (the used initial state values are shown in Fig. \ref{Fig:InitialState}).


The inventory loss (waste) values $w_{i,k}(t)$ over $t\in\N_T$ at each $i \in \N_N, k \in \N_n$ were generated using random normal distributions $\mathcal{N}(\overline{w}_{i,k},0.2\overline{w}_{i,k})$ were $\overline{w}_{i,k} \triangleq 10 + 2 \times \text{randi}(1,5) + 2i$. Similarly, the transportation loss (waste) values $\bar{w}_{i,k}(t)$ over $t\in\N_T$ at each $i \in \N_N, k \in \N_n$ were generated using random normal distributions $\mathcal{N}(\overline{\bar{w}}_{i,k},0.2\overline{\bar{w}}_{i,k})$ were $\overline{\bar{w}}_{i,k} \triangleq 10 + 2\times \text{randi}(1,5) + 2i$. To make the time variation of these randomly generated loss profiles $w_{i,k}(t)$ and $\bar{w}_{i,k}(t)$ more realistic, we used exponential smoothing (e.g., $w_{i,k}^{smooth}(t) = \alpha w_{i,k}(t) + (1-\alpha) w_{i,k}(t-1)$) with the smoothing parameter $\alpha = 0.5$.

The demand values $d_{in}(t)$ over $t\in\N_T$ at each $i\in\N_N$ were generated using random normal distributions  $\mathcal{N}(\overline{d}_{in},0.2\overline{d}_{in})$. Here, the mean value $\overline{d}_{in}$ depends on the current ``day of the week,'' and is selected from seven prespecified values randomly generated using the formula $100 + 2\times\text{randi}(1,10N) + 20i$ as shown in Fig. \ref{Fig:MeanDemands}. The complete demand profiles generated using this strategy are shown in Fig. \ref{Fig:Demands}. Similar to before, to make the time variation of the randomly generated demand profiles $d_{in}(t)$ more realistic, we used exponential smoothing with the smoothing parameter $\alpha = 0.1$. It is worth noting that, when evaluating different SCN control strategies, for their steady-state control components  \eqref{Eq:SteadyStateControlLink}, we only used the knowledge of overall mean demand values shown in dashed lines in Figs. \ref{Fig:MeanDemands}-\ref{Fig:Demands}.

\begin{figure}[!t]
    \centering
    \includegraphics[width=1\linewidth]{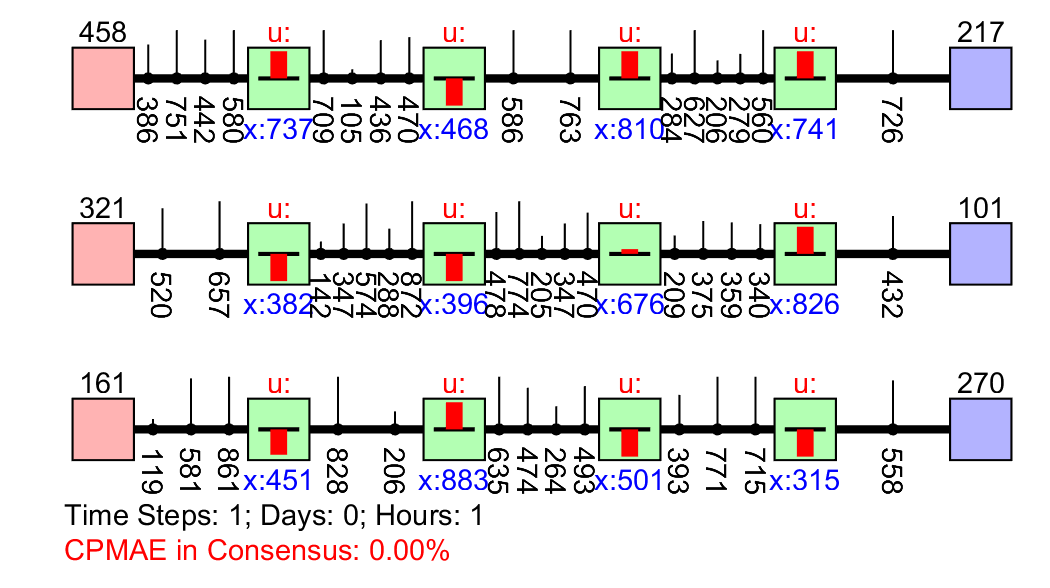}
    \caption{Initial state of the considered SCN: Suppliers (red), inventories (green), demands (blue), links (black).}
    \label{Fig:InitialState}
\end{figure}

\begin{figure}[!t]
\centering
\begin{minipage}{\columnwidth}
    \centering
    \includegraphics[width=0.6\linewidth]{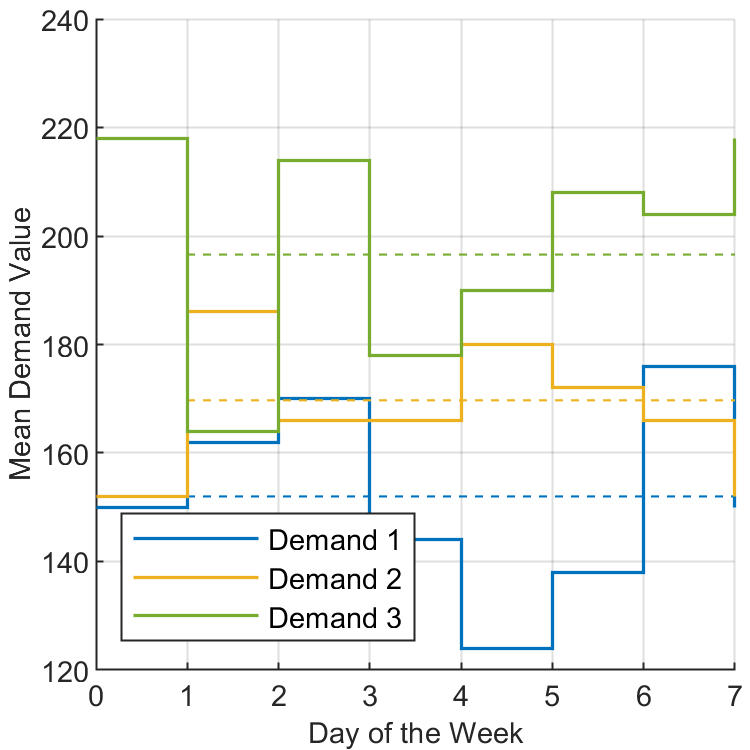}
    \caption{The daily mean customer demand values in a week used for simulating the customer demands (straight lines), and the overall mean customer demand values $\{\overline{d}_{in}: i\in\N_N\}$ used in steady-state controllers  \eqref{Eq:SteadyStateControlLink} (dashed lines).}
    \label{Fig:MeanDemands}
\end{minipage}
\hfill
\begin{minipage}{\columnwidth}
    \centering
    \includegraphics[width=0.6\linewidth]{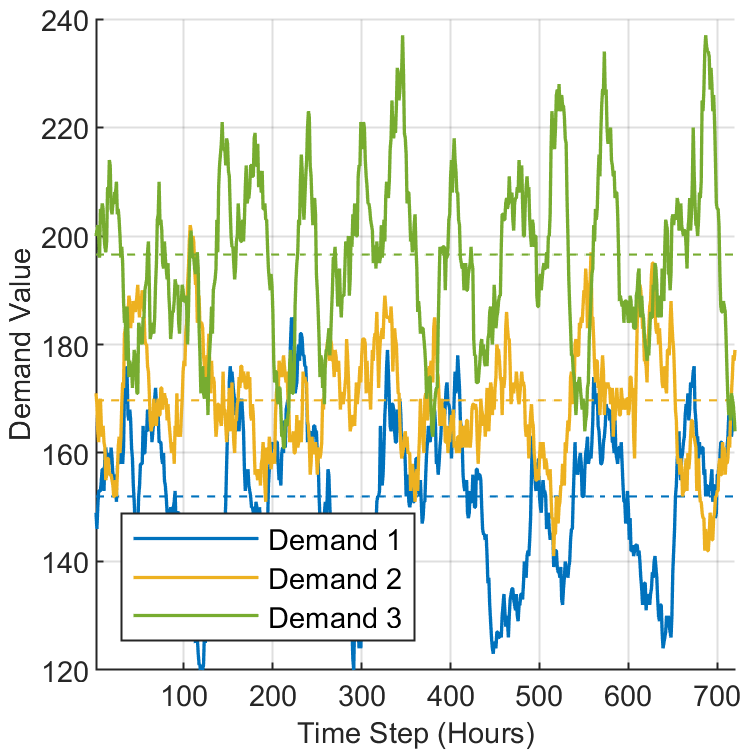}
    \caption{The simulated customer demand values $\{d_{in}(t):i\in\N_N, t\in\N_T\}$ applied to the considered SCN.}
    \label{Fig:Demands}
\end{minipage}
\end{figure}

\subsection{Evaluated SCN Control Strategies}

We considered five different SCN control strategies: 
(i) Local steady-state control (LSSC),
(ii) Local state-feedback control (LSFC),
(iii) Global consensus control (GCC)
(iV) Distributed consensus control under hard topology constraints (DCC-C), and 
(v) Distributed consensus control under soft topology constraints (DCC-U). 
Under the first two methods (LSSC and LSFC), inventories use only local information for their respective supply chains to make control decisions. The third method (GCC) is a generic global consensus control approach inspired by \cite{Oltafi-Saber2004}, and thus, inventories may use complete inventory state information of the SCN. The last two methods (DCC-C and DCC-U) limit inventory state information sharing over the SCN using the proposed distributed consensus control and communication topology co-design approach in this paper.

More specifically, in the LSSC method, inventories place fixed orders based on their downstream mean demand and waste values using the steady-state control law $\overline{u}_i$ \eqref{Eq:SteadyStateControlLink} designed in Prop. \ref{Pr:SteadyStateControl} with $\tilde{u}_i(t) = \0$ for all $i\in \N_N, t\in\N_T$ in \eqref{Eq:SteadyStateControl0}. Besides, the LSFC method adds to the LSSC method the local state-feedback control law $\tilde{u}_i(t)$ \eqref{Eq:LocalStateFeedbackController} with $\tilde{\tilde{u}}_i(t) = \0$ for all $i\in\N_N, t\in\N_T$. To design the local controller gains $\{\mathcal{L}_i, i \in \N_N\}$ required in this LSFC method, we used the dissipativity-based approach given in Prop. \ref{Th:LocalControlDesign}. On the other hand, the GCC method adds to the LSSC method a global consensus controller $\tilde{\tilde{u}}_i(t)$ of the form \eqref{Eq:ConsensusControl} with $\tilde{u}_i(t) = \0$ for all $i\in\N_N, t\in\N_T$. However, in GCC design, unlike in the proposed DCC co-design, the communication topology cannot be constrained and/or optimized, and hence, the GCC method requires a fully connected communication topology, as shown in Fig. \ref{Fig:TopologyGCC}.

The proposed DCC methods use: (i) the steady-state control law $\overline{u}_i$ \eqref{Eq:SteadyStateControlLink} designed in Prop. \ref{Pr:SteadyStateControl}, (ii) the local state-feedback control law $\tilde{u}_i(t)$ \eqref{Eq:LocalStateFeedbackController} designed in Th. \ref{Th:ImprovedLocalControlDesign}, and (iii) the distributed consensus control law $\tilde{\tilde{u}}_i(t)$ \eqref{Eq:ConsensusControl} designed in Th. \ref{Th:GlobalControlDesign}, for all $i\in\N_N, t\in\N_T$. Both DCC methods used a reference topology as shown in Fig. \ref{Fig:RefTop} for the selection of the cost coefficients 
$C \triangleq [[C_{ij}^{kl}]_{k,l\in\N_nj}]_{i,j\in\N_N}$ in \eqref{Eq:Th:GlobalControlDesign}. The distinction between the DCC variants, DCC-C and DCC-U, lies in how the resulting communication topology was constrained when solving the proposed co-design problem in Th. \ref{Th:GlobalControlDesign}. In particular, the DCC-C method included some additional constraints in \eqref{Eq:Th:GlobalControlDesign} (of the form $\bar{\mathcal{K}}_{ij}^{kl}=0$ for some $i,j\in\N_N$ and $k,l\in\N_n$) to constrain the resulting topology to use only the communication links in the reference topology shown in Fig. \ref{Fig:RefTop}. In contrast, the DCC-U method did not use such additional constraints. 

For the considered SCN shown in Fig. \ref{Fig:InitialState}, under the reference topology shown in Fig. \ref{Fig:RefTop}, the obtained optimal communication topologies in DCC-C and DCC-U methods are shown in Figs. \ref{Fig:TopologyHard} and \ref{Fig:TopologySoft}, respectively.

\begin{figure}[t!]
\centering
\begin{minipage}{0.48\textwidth}
\centering
\includegraphics[width=0.6\linewidth]{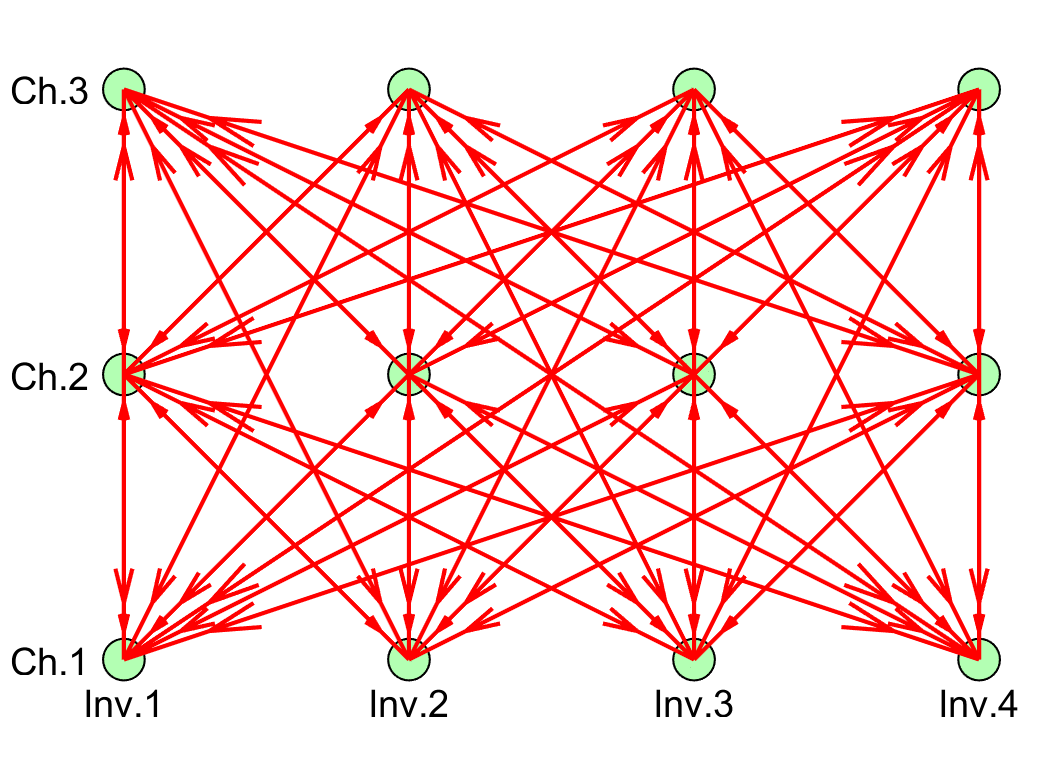}
\caption{Fully connected communication topology used in the conventional global consensus control (GCC) strategy.}
\label{Fig:TopologyGCC}
\end{minipage}
\hfill
\begin{minipage}{0.48\textwidth}
    \centering
    \includegraphics[width=0.6\linewidth]{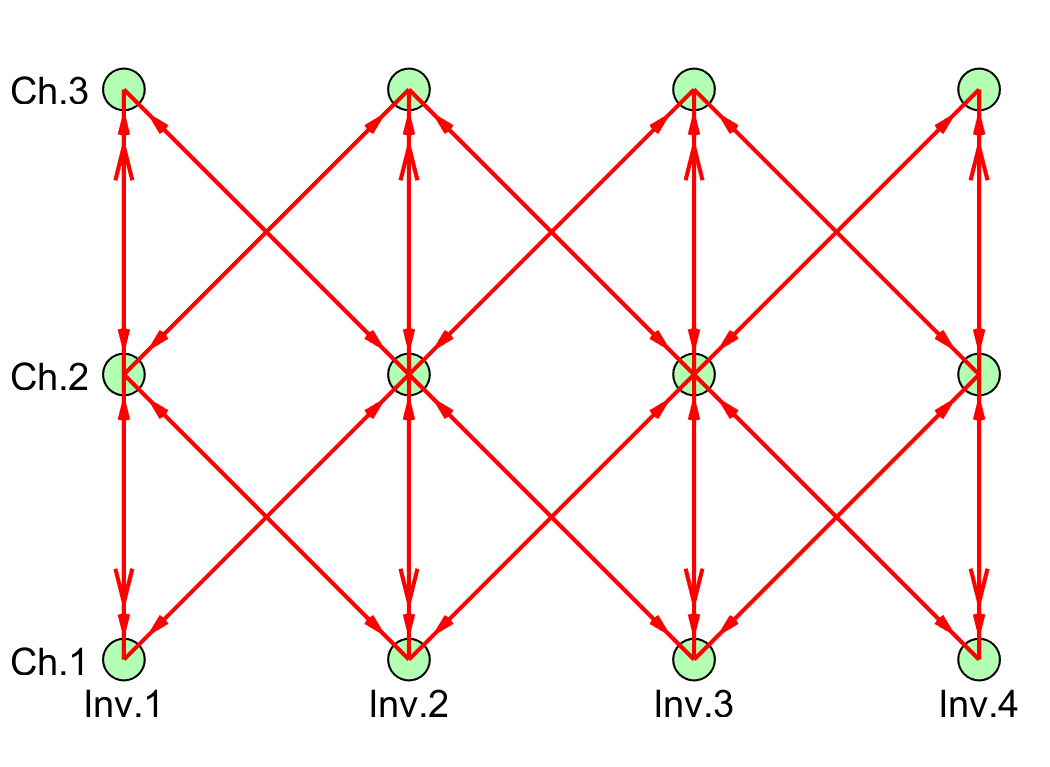}
    \caption{Reference topology used for selecting the cost coefficients $c \triangleq [c_{ij}]_{i,j\in\N_N}$ \eqref{Eq:Th:GlobalControlDesign} in the proposed co-design method for DCC-C and DCC-U strategies.}
    \label{Fig:RefTop}
\end{minipage}
\\
\begin{minipage}{0.48\textwidth}
    \centering
    \includegraphics[width=0.6\linewidth]{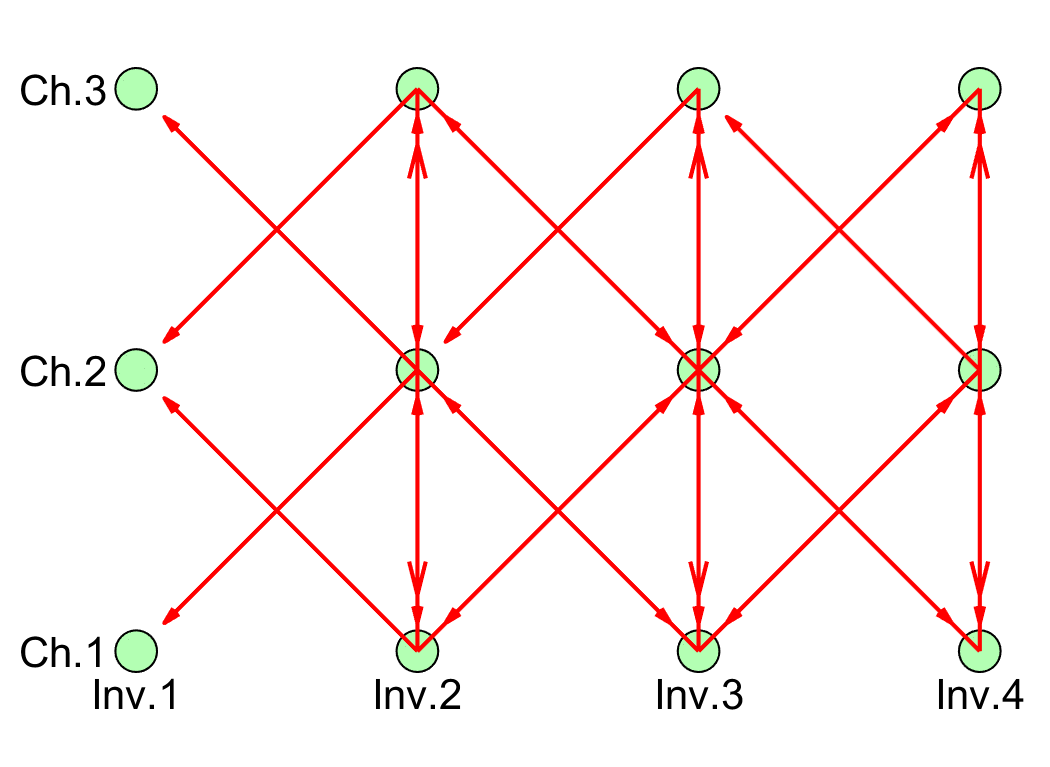}
    \caption{Obtained optimal topology from the proposed co-design method for the DCC-C strategy.}
    \label{Fig:TopologyHard}
\end{minipage}
\hfill
\begin{minipage}{0.48\textwidth}
    \centering
\includegraphics[width=0.6\linewidth]{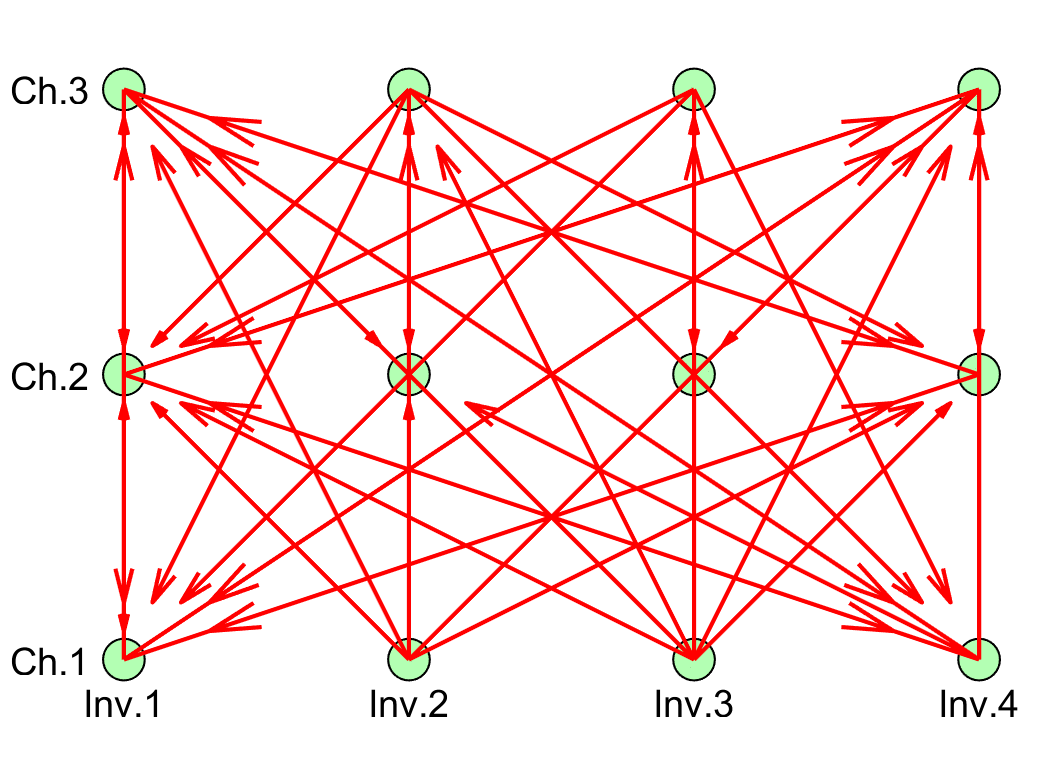}
    \caption{Obtained optimal topology from the proposed co-design method for the DCC-U strategy.}
    \label{Fig:TopologySoft}
\end{minipage}
\end{figure}

\subsection{Observations}

To evaluate the consensus performance of different SCN control methods, we used the percentage mean absolute error (PMAE) in consensus metric defined as 
$$
\text{PMAE}\,(t) \triangleq \left[\frac{1}{nN} \Vert z(t)\Vert_1 \right] \times \frac{100}{\overline{X}} \%
$$
where $z(t) \triangleq [z_i(t)]_{i\in\N_N}$ is the observed performance outputs vector defined in \eqref{Eq:ConsensusPerformance}-\eqref{Eq:SupplyChainNetworkInterconnection}, and $\overline{X}$ was chosen as $\overline{X} \triangleq 500$ inspired by the used desired inventory levels choice $\overline{x}_{i,k} = 500, \forall i\in\N_N,k\in\N_n$.   
By averaging the observed PMAE profiles over a large number of random realizations (Monte Carlo simulations), we also computed the Average PMAE (APMAE) metric (i.e., $\text{APMAE}(t) \triangleq \mathbb{E}(\text{PMAE}(t))$). On the other hand, by averaging the observed PMAE and APMAE profiles over time, we also computed the cumulative PMAE (CPMAE) and cumulative APMAE (CAPMAE) metrics, respectively (e.g.,  
$\text{CPMAE}(t) 
\triangleq \frac{1}{t}\sum_{\tau\in\N_t} \text{PMAE}(t)$).

When simulating the considered SCN (at each realization), besides the randomly generated quantities like inventory/transportation losses and customer demands, we also randomly generated initial conditions (at $t=1$), transportation failures (at $t = 240$), and inventory failures (at $t = 480$). In particular, transportation failures occur at two randomly selected transportation links in the SCN where each selected (``failed'') link will lose all of its carrying products. Similarly, inventory failures occur at four randomly selected inventories in the SCN, where each selected (``failed'') inventory will lose all of its stored products.

The observed PMAE, APMAE and CAPMAE profiles under different considered SCN control methods are shown in Figs. \ref{Fig:ConsensusPerformance}, \ref{Fig:AverageConsensusPerformance} and \ref{Fig:AverageCumulativeConsensusPerformance}, respectively. The final CAPMAE values (obtained from Fig. \ref{Fig:AverageCumulativeConsensusPerformance}) are summarized in Tab. \ref{Tab:ConsensusPerformance}.

\begin{figure*}[t]
\centering
\begin{minipage}{0.32\textwidth}
    \centering
    \includegraphics[width=\linewidth]{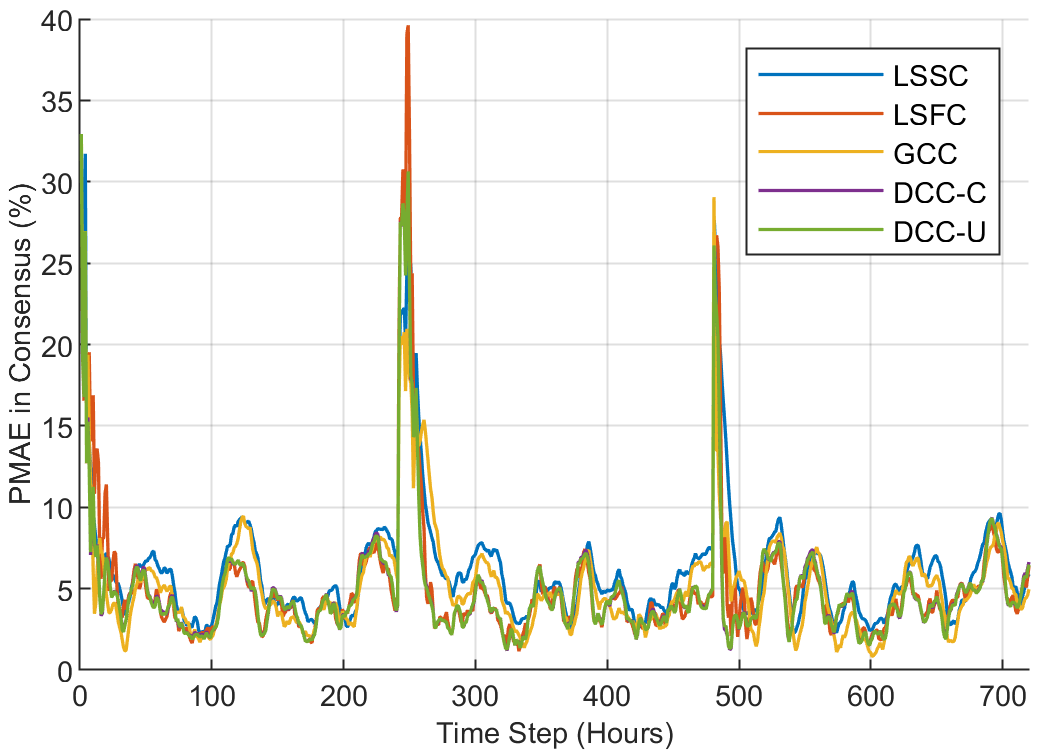}
\caption{Consensus performance (PMAE in consensus) observed under different SCN control strategies.}
\label{Fig:ConsensusPerformance}
\end{minipage}
\hfill
\begin{minipage}{0.32\textwidth}
    \centering
    \includegraphics[width=\linewidth]{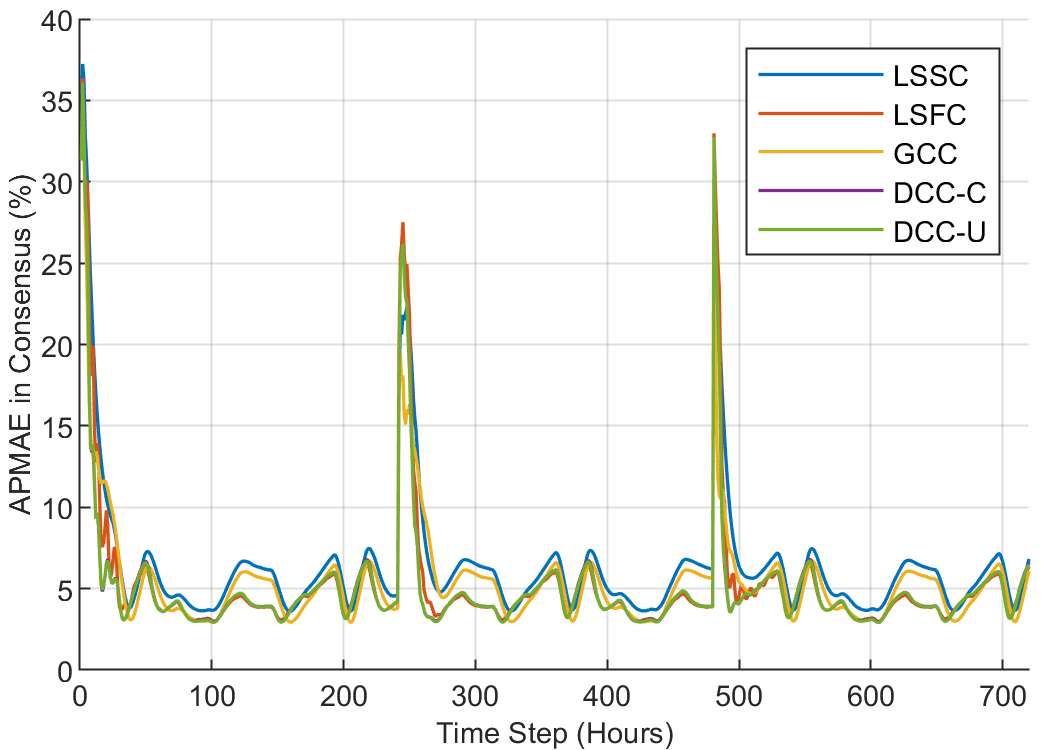}
\caption{Average PMAE in consensus (over 1000 realizations) observed under different SCN control strategies.}
\label{Fig:AverageConsensusPerformance}
\end{minipage}
\hfill
\begin{minipage}{0.32\textwidth}
    \centering
    \includegraphics[width=\linewidth]{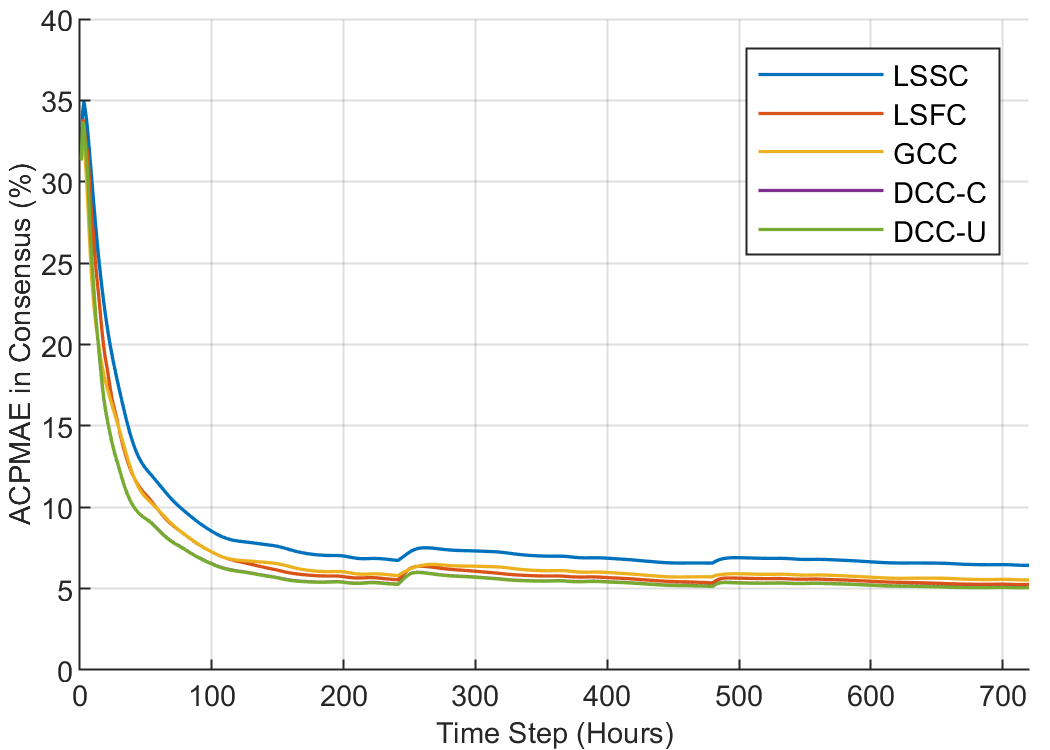}
\caption{Average PCMAE in consensus (over 1000 realizations) observed under different SCN control strategies.}
\label{Fig:AverageCumulativeConsensusPerformance}
\end{minipage}
\end{figure*}

\begin{figure}[!t]
\centering
    \includegraphics[width=1\linewidth]{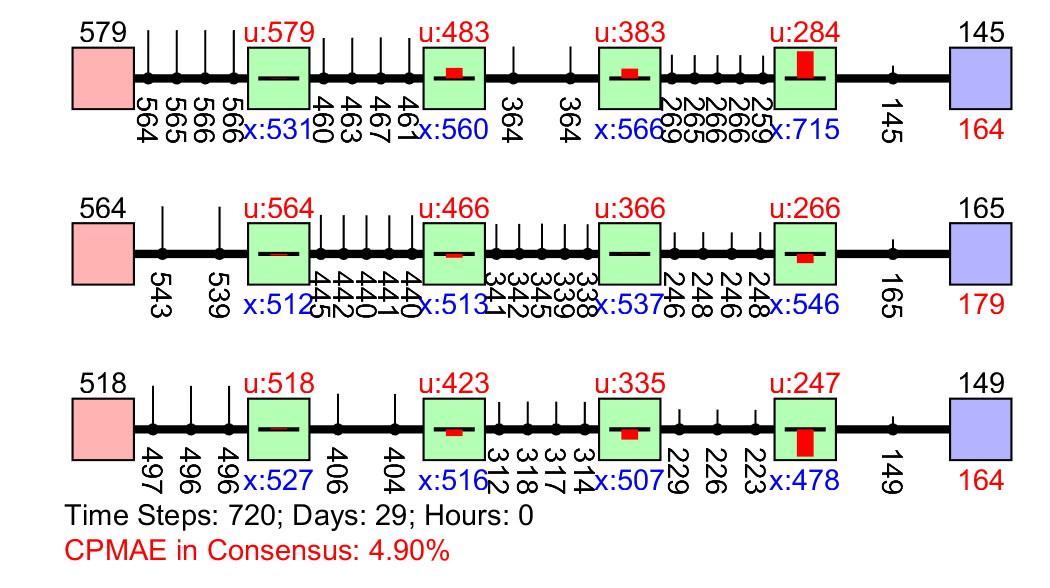}
    \caption{Final state of the SCN under the proposed SCN control strategy (i.e., the DCC-U strategy).}
    \label{Fig:FinalStateDCSGC}
\end{figure}

\begin{table}[!t]
    \centering
    \caption{Overall consensus performance observed under different SCN control methods.}
    \begin{tabular}{|c|c|} \hline\
    SCN Control Method   & Final CAPMAE Value (\%)\\ \hline
      LSSC      & 6.432 \\  \hline
      LSFC      & 5.256 \\ \hline
      GCC       & 5.532 \\ \hline
      DCC-C     & 5.070 \\ \hline
      DCC-U     & \textbf{5.059} \\ \hline
    \end{tabular}
    \label{Tab:ConsensusPerformance}
\end{table}

\subsection{Discussion}

From the reported observations, it is clear that the proposed DCC methods outperform conventional methods like LSSC, LSFC, and GCC in terms of maintaining a consensus among the supply chains in the SCN - in the presence of considerable uncertainties and failures. It is worth highlighting that the proposed DCC-U method provides the fastest convergence to a consensus from the applied diverse initial conditions and injected transportation and inventory failures. The final state of the SCN observed under the proposed DCC-U strategy in a simulation is shown in Fig. \ref{Fig:FinalStateDCSGC}. 

The observed superior consensus control performance of the DCC-U method is partly due to the specifically designed communication topology shown in Fig. \ref{Fig:TopologySoft}. In this communication topology, several key features can be observed: 
\begin{enumerate}
    \item It is suboptimal to use all the communication links. 
    \item More information is required at the inventories located at the corners of the SCN as opposed to those located at the middle of the SCN.
    \item Communication links between non-adjacent inventories are preferred as opposed to those between adjacent inventories.
\end{enumerate}
These observations provide crucial yet non-trivial intuitions regarding suitable communication topologies for distributed consensus control in SCNs to minimize the adverse bullwhip and ripple effects. 


\section{Conclusion}\label{Sec:Conclusion}
In this paper, we developed a robust distributed inventory consensus control solution for supply chain networks. Maintaining such a consensus among inventories synchronizes parallel supply chains in the SCN while enhancing coordination and robustness of the overall SCN, mitigating both bullwhip and ripple effects.
For this purpose, we modeled the SCN as a networked system comprised of dynamic subsystems with some disturbance inputs and performance outputs. We then used the dissipativity theory to develop a co-design strategy for distributed consensus controllers and their underlying communication topology. The proposed dissipativity-based approach, due to its omission of exact dynamic models, provides a co-design that is inherently robust to disturbances. Moreover, the formulated co-design problem is convex, enabling efficient solutions. The resulting consensus control solutions outperformed the conventional static steady-state control, local dynamic control, as well as global consensus control strategies. Moreover, the resulting optimal communication topologies provide useful insights into effective communication topologies for consensus control of SCNs. Future work is directed towards including the design of adaptive control laws to estimate and update the involved consensus controller parameters like the mean transportation and inventory loss values and mean customer demand values.

\appendices
\section{Proof of the Proposition \ref{Pr:LTISystemXDisspativity}}
\label{App:Pr:LTISystemXDisspativity}

Let $\tilde{x}(t) \triangleq x(t) - x^*$, $\tilde{u}(t) \triangleq u(t) - u^*$ and $\tilde{y}(t) \triangleq y(t) - y^*$, where $x^* \in \mathcal{X}$ is an equilibrium state of \eqref{Eq:Pr:LTISystemXDisspativity1}, and $u^*$ and $y^*$ represent a corresponding input and output combination. Consequently, $x^*,u^*$ and $y^*$ satisfy
$$
x^* = A x^* + B u^*, \quad \mbox{ and } \quad
y^* = Cx^* +Du^*,    
$$
which, together with \eqref{Eq:Pr:LTISystemXDisspativity1} give the ``error'' dynamics:
\begin{equation}\label{Eq:Pr:LTISystemXDisspativity1Step1}
\begin{aligned}
\tilde{x}(t+1) = A\tilde{x}(t) + B\tilde{u}(t),\\
\tilde{y}(t) = C\tilde{x}(t) + D\tilde{u}(t).
\end{aligned}
\end{equation}

For \eqref{Eq:Pr:LTISystemXDisspativity1}, consider the candidate quadratic storage function 
\begin{equation}
    V(x(t),x^*) \equiv V(\tilde{x}(t)) \triangleq \tilde{x}^\T(t) P \tilde{x}(t),
\end{equation}
where $P$ is to be determined such that $P>0$. Using \eqref{Eq:Pr:LTISystemXDisspativity1Step1} we now can evaluate 
$\Delta V(\tilde{x}(t)) \triangleq V(\tilde{x}(t+1)) - V(\tilde{x}(t))$ as
\begin{align}
\Delta V 
=&\ \left(\bm{A & B} \bm{\tilde{x}(t) \\ \tilde{u}(t)}\right)^\T P \, \bigstar - \tilde{x}^\T(t) P \tilde{x}(t) \nonumber \\
=&\ \bm{\tilde{x}(t) \\ \tilde{u}(t)}^\T \bm{V_1-V_2} \bm{\tilde{x}(t) \\ \tilde{u}(t)}, 
\label{Eq:Pr:LTISystemXDisspativity1Step2}
\end{align}
where 
\begin{equation*}
\begin{aligned}
V_1 \triangleq  \bm{A & B}^\T P \bm{A & B} \mbox{ and }
V_2 \triangleq \bm{\I & \0}P \bm{\I & \0}.    
\end{aligned}
\end{equation*}

For the $X$-EID of \eqref{Eq:Pr:LTISystemXDisspativity1}, we require the supply rate function $s(\tilde{u}(t),\tilde{y}(t)) \triangleq \bm{\tilde{u}^\T(t) &  \tilde{y}^\T(t)} X \bm{\tilde{u}^\T(t) &  \tilde{y}^\T(t)}^\T$. Using \eqref{Eq:Pr:LTISystemXDisspativity1Step1} we can simplify this storage function $s$ into the form
\begin{align}
s =&\   
\bm{\tilde{y}(t) \\ \tilde{u}(t)}^\T 
\bm{X^{22} & X^{21} \\ X^{12} & X^{11}} \bm{\tilde{y}(t) \\ \tilde{u}(t)} \nonumber\\
=&\ \left(\bm{C & D \\ \0 & \I}\bm{\tilde{x}(t) \\ \tilde{u}(t)}\right)^\T \bm{X^{22} & X^{21} \\ X^{12} & X^{11}} \bigstar \nonumber\\
=& \bm{\tilde{x}(t) \\ \tilde{u}(t)}^\T S \bm{\tilde{x}(t) \\ \tilde{u}(t)},
\label{Eq:Pr:LTISystemXDisspativity1Step3}
\end{align}
where
$$ 
\scriptsize 
S \triangleq \bm{C^\T X^{22} C & C^\T X^{22} D \\ D^\T X^{22}C + X^{12}C & D^\T X^{22}D + X^{12}D + D^\T X^{21} + X^{11}}.$$ 

Using \eqref{Eq:Pr:LTISystemXDisspativity1Step2} and \eqref{Eq:Pr:LTISystemXDisspativity1Step3}, the required differential dissipativity inequality (see Def. \ref{Def:EID}) can now be expressed as 
\begin{equation} \label{Eq:Pr:LTISystemXDisspativity1Step4}
\Delta V \leq s \iff S-(V_1-V_2) \geq 0.
\end{equation}
As shown in \cite[Lm. 1]{WelikalaP52022}, for any $\Gamma, \Phi, \Theta$ such that $\Gamma = \Gamma^\T$ and $\Theta > 0$: 
\begin{equation}
\label{Eq:Pr:LTISystemXDisspativity1Step5}
\Gamma - \Phi^\T \Theta \Phi > 0 \iff \bm{\Theta & \Theta \Phi \\ \Phi^\T \Theta & \Gamma} > 0.
\end{equation}
Using this result with $\Gamma \triangleq (S+V_2)$, $\Phi \triangleq \bm{A & B}$ and $\Theta = P$ (note that $\Phi^\T \Theta \Phi = V_1$), we can obtain equivalent condition for \eqref{Eq:Pr:LTISystemXDisspativity1Step4} in the form  \eqref{Eq:Pr:LTISystemXDisspativity2}. This completes the proof.


\section{Proof of Corollary \ref{Co:LTISystemXDisspativation}}
\label{App:Co:LTISystemXDisspativation}


In Prop. \eqref{Pr:LTISystemXDisspativity}, replacing $A$, $B$ and $D$ respectively with $A+BL$, $\I$ and $\0$, we can obtain its main LMI constraint as
$$\Gamma \triangleq 
\scriptsize \bm{
P & PA+PBL & P \\
A^\T P + L^\T B^\T P & P + C^\T X^{22} C & C^\T X^{21} \\
P & X^{12} C  & X^{11}
} 
\normalsize \geq 0$$
This can be written as 
$\Gamma - \Phi^\T \Theta \Phi \geq 0$ 
where 
\begin{equation*}
\begin{aligned}
\Gamma \triangleq&\ 
\scriptsize \bm{
P & PA+PBL & P \\
A^\T P + L^\T B^\T P & P & C^\T X^{21} \\
P & X^{12} C  & X^{11}
} \normalsize, \\
\Phi \triangleq&\ \bm{\0 & C & \0}, \quad \mbox{ and } \quad 
\Theta \triangleq \bm{-X_{22}}.
\end{aligned}
\end{equation*}
Using \eqref{Eq:Pr:LTISystemXDisspativity1Step5}, we next can obtain an equivalent condition as
$$
\scriptsize
\bm{ 
(-X^{22})^{-1} & \0 & C & \0 \\
\0 & P & PA+PBL & P \\
C^\T & A^\T P + L^\T B^\T P & P  & C^\top X^{21}\\
\0 & \I & X^{12} C  & X^{11}} 
\normalsize  \geq 0.
$$
Using the congruence concept, by pre- and post- multiplying with $\diag(\bm{\I & P^{-1} & P^{-1} & \I})$ and subsequently changing variables using $Q \triangleq P^{-1}$, we get
$$
\scriptsize
\bm{ 
(-X^{22})^{-1} & \0 & CQ & \0 \\
\0 & Q & AQ+BLQ & \I \\
QC^\T  & QA^\T + QL^\T B^\T & Q  & QC^\top X^{21}\\
\0 & \I & X^{12} CQ  & X^{11}} 
\normalsize  \geq 0.
$$
Finally, using the change of variables $P \triangleq Q$ and $K \triangleq LP$ (respectively), we can obtain an equivalent LMI condition in the form  
\eqref{Eq:Co:LTISystemXDisspativation2}, which completes the proof.

\section{Proof of Proposition \ref{Pr:NSC2Synthesis}}
\label{App:Pr:NSC2Synthesis}


For the discrete-time networked system $\Sigma$, consider an equilibrium state $x^* \in \mathcal{X}$ where the corresponding networked system equilibrium input and output are $w^*$ and $z^*$, respectively. As $x^* = [x_i^*]_{i\in\N_N}$ where each $x_i^* \in \mathcal{X}_i$, the corresponding subsystem equilibrium inputs and outputs are concatenated in $u^* = [u_i^*]_{i\in\N_N}$ and $y^* = [y_i^*]_{i\in\N_N}$, respectively. Recall that there exists an implicit relationship between $w^*, z^*, u^*$ and $y^*$ such that:
\begin{equation}
\label{Eq:Pr:NSC2SynthesisStep1}
\bm{u^* \\ z^*} = \bm{M_{uy} & M_{uw} \\ M_{zy} & M_{zw}} \bm{y^* \\ w^*}.
\end{equation}
Let 
$\tilde{x}(t) = x(t) - x^*$, 
$\tilde{w}(t) = w(t) - w^*$, 
$\tilde{z}(t) = z(t) - z^*$, 
$\tilde{u}(t) = u(t) - u^*$ and 
$\tilde{y}(t) = y(t) - y^*$.
Using \eqref{Eq:NSC2Interconnection} and \eqref{Eq:Pr:NSC2SynthesisStep1}, it is easy to show that
\begin{equation}
\label{Eq:Pr:NSC2SynthesisStep2}
\bm{\tilde{u}(t) \\ \tilde{z}(t)} = \bm{M_{uy} & M_{uw} \\ M_{zy} & M_{zw}} \bm{\tilde{y}(t) \\ \tilde{w}(t)}.
\end{equation}

We propose the storage function 
\begin{equation}
    \label{Eq:Pr:NSC2SynthesisStep3}
V(x(t),x^*) \triangleq \sum_{i\in\N_N} p_i V_i(x_i(t),x_i^*)    
\end{equation}
for the networked system, where each $p_i>0$ is a scalar parameter and $V_i(x_i(t),x_i^*)$ is the storage function of the subsystem $\Sigma_i, i\in\N_N$. Let 
$\Delta V \triangleq V(x(t+1),x^*) - V(x(t),x^*)$ and 
$\Delta V_i \triangleq V_i(x_i(t+1), x_i^*) - V_i(x_i(t+1),x_i^*), \forall i\in\N_N$. Using \eqref{Eq:Pr:NSC2SynthesisStep3}, we get:
\begin{align*}
    \Delta V &= \sum_{i\in\N_N} p_i \Delta V_i \\
    &\leq \sum_{i\in\N_N}   \bm{\tilde{u}_i(t) \\ \tilde{y}_i(t)}^\T p_i X_i \bm{\tilde{u}_i(t) \\ \tilde{y}_i(t)}\\ 
    &= \bm{\tilde{u}(t) \\ \tilde{y}(t)}^\T X_p \bm{\tilde{u}(t) \\ \tilde{y}(t)},
\end{align*} 
where the second step is a consequence of each subsystem $\Sigma_i, i\in \N_N$ being $X_i$-EID and the last step is a result of vectorization, with $X_p \triangleq [X_p^{kl}]_{k,l\in\N_2}$ and $X_p^{kl} \triangleq \diag[p_i X_i^{kl}]_{i\in\N_N}, \forall k,l\in\N_2$. Subsequently, using \eqref{Eq:Pr:NSC2SynthesisStep2}, we can obtain
\begin{equation}
    \label{Eq:Pr:NSC2SynthesisStep4}
    \Delta V \leq \left(\bm{M_{uy} & M_{yw}\\ \I & \0}\bm{\tilde{y}(t)\\ \tilde{w}(t)}\right)^\T X_p \bigstar   
\end{equation}

For the $\textbf{Y}$-EID of the networked system, we require the difference dissipativity inequality $\Delta V \leq s(\tilde{w}(t),\tilde{z}(t))$ to hold, where 
$s(\tilde{w}(t),\tilde{z}(t)) \triangleq \bm{\tilde{w}^\T(t) &  \tilde{z}^\T(t)} \textbf{Y} \bm{\tilde{w}^\T(t) &  \tilde{z}^\T(t)}^\T$ denotes the supply rate function. Note that, using \eqref{Eq:Pr:NSC2SynthesisStep2} we can simplify this supply rate function $s$ into the form
\begin{align}
s =&\   
\bm{\tilde{z}(t) \\ \tilde{w}(t)}^\T 
\bm{\textbf{Y}^{22} & \textbf{Y}^{21} \\ \textbf{Y}^{12} & \textbf{Y}^{11}} \bm{\tilde{z}(t) \\ \tilde{w}(t)} \nonumber\\
=&\ \left(\bm{M_{zy} &  M_{zw} \\ \0 & \I}\bm{\tilde{y}(t) \\ \tilde{w}(t)}\right)^\T \bm{\textbf{Y}^{22} & \textbf{Y}^{21} \\ \textbf{Y}^{12} & \textbf{Y}^{11}}  \bigstar 
\label{Eq:Pr:NSC2SynthesisStep5}
\end{align}

Using \eqref{Eq:Pr:NSC2SynthesisStep4} and \eqref{Eq:Pr:NSC2SynthesisStep5}, a sufficient condition for $\Delta V \leq s$ (i.e., $\textbf{Y}$-EID of the networked system) can be identified as 
\begin{align}
&\bm{M_{uy} & M_{yw}\\ \I & \0}^\T X_p \bigstar   
\leq 
\bm{M_{zy} &  M_{zw} \\ \0 & \I}^\T \bm{\textbf{Y}^{22} & \textbf{Y}^{21} \\ \textbf{Y}^{12} & \textbf{Y}^{11}}  \bigstar  \nonumber\\
&\iff 
0 \leq \Gamma - \bm{M_{uy} & M_{yw}}^\T X_p^{11}\bigstar \\
&- \bm{M_{zy} &  M_{zw}}^\T (-\textbf{Y}^{22}) \bigstar
\label{Eq:Pr:NSC2SynthesisStep6}
\end{align}
where 
$$
\Gamma \triangleq
\scriptsize
\bm{-M_{uy}^\T X_p^{12} - X_p^{21}M_{uy} - X_p^{22} & -X_p^{21}M_{yw} + M_{zy}^\T \textbf{Y}^{21} \\  - M_{yw}^\T X_p^{12} + \textbf{Y}^{12}M_{zy} & 
M_{zw}^\T \textbf{Y}^{21} + \textbf{Y}^{12}M_{zw} + \textbf{Y}^{11}}
\normalsize.
$$
Finally, under Assumptions \ref{As:NegativeDissipativity} and \ref{As:PositiveDissipativity}, using the technical result \eqref{Eq:Pr:LTISystemXDisspativity1Step5} and change of variables $\bm{L_{uy} & L_{uw}} \triangleq X_p^{11}\bm{M_{uy} & M_{uw}}$, we can obtain equivalent conditions for \eqref{Eq:Pr:NSC2SynthesisStep6} as given in \eqref{Eq:Pr:NSC2Synthesis2} and \eqref{Eq:Pr:NSC2Synthesis2Alternative}. This completes the proof. 


\bibliographystyle{IEEEtran}
\bibliography{References}

\begin{IEEEbiography}[{\includegraphics[width=1in,height=1.25in,clip,keepaspectratio]{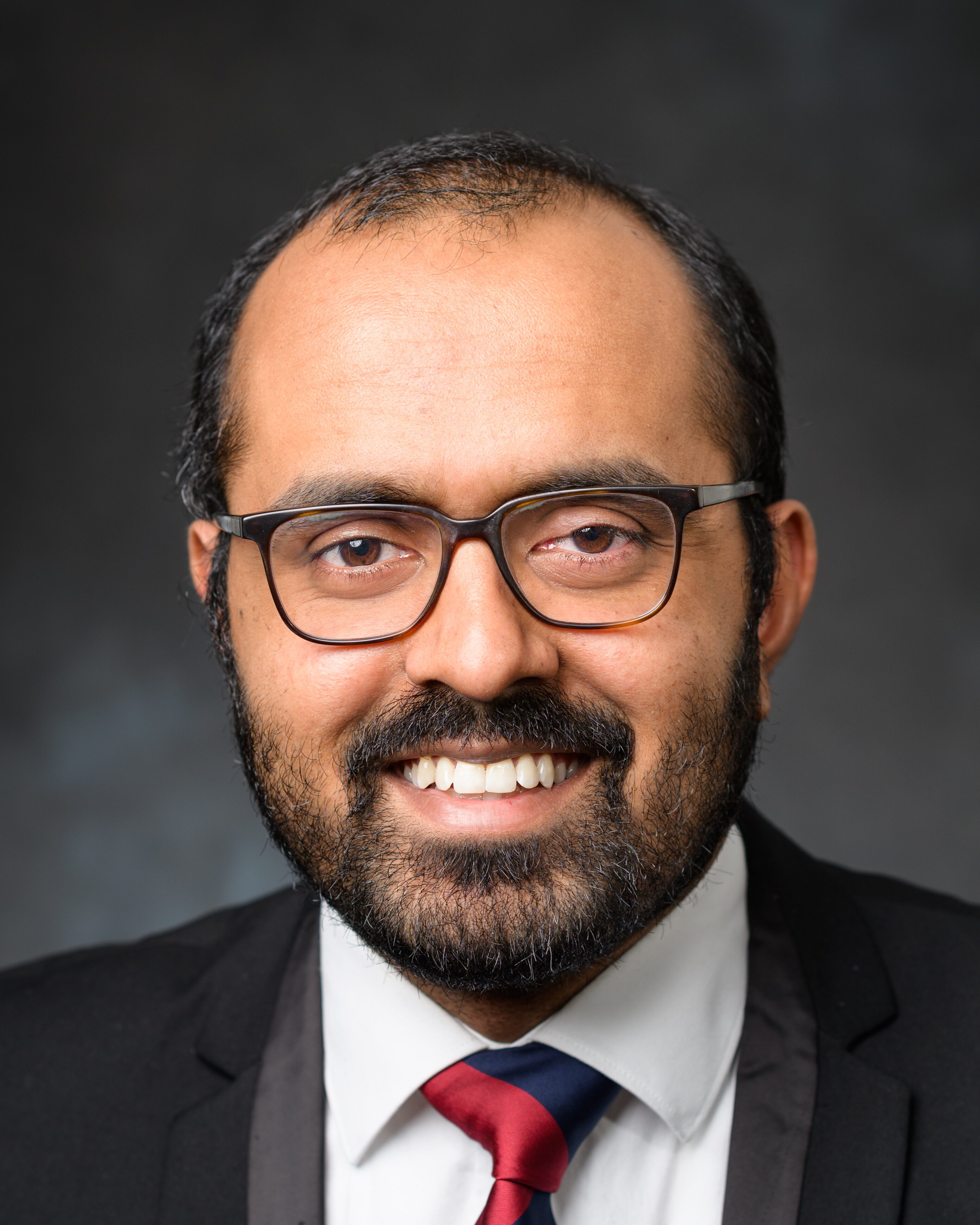}}]{Shirantha Welikala} (Member, IEEE) 
received the B.Sc. degree in electrical and electronic engineering from University of Peradeniya, Peradeniya, Sri Lanka, in 2015 and the M.Sc. and the Ph.D. degrees in systems engineering from Boston University, Brookline, MA, USA, in 2019 and 2021, respectively. From 2015 to 2017, he was a Temporary Instructor/Research Assistant in the Department of Electrical and Electronic Engineering, University of Peradeniya, Sri Lanka. From 2021 to 2023, he was a Postdoctoral Research Fellow in the Department of Electrical Engineering, University of Notre Dame, Notre Dame, IN, USA. He is currently an Assistant Professor in the Department of Electrical and Computer Engineering, Stevens Institute of Technology, Hoboken, NJ, USA. His main research interests include control and optimization of cooperative multi-agent systems, control of networked systems, passivity, machine-learning, robotics, and smart grid. He is a recipient of several awards, including the 2015 Ceylon Electricity Board Gold Medal, the 2019 and 2023 President's Awards for Scientific Research in Sri Lanka, and the 2021 Outstanding Ph.D. Dissertation Award in Systems Engineering.
\end{IEEEbiography}

\begin{IEEEbiography}[{\includegraphics[width=1in,height=1.25in,clip,keepaspectratio]{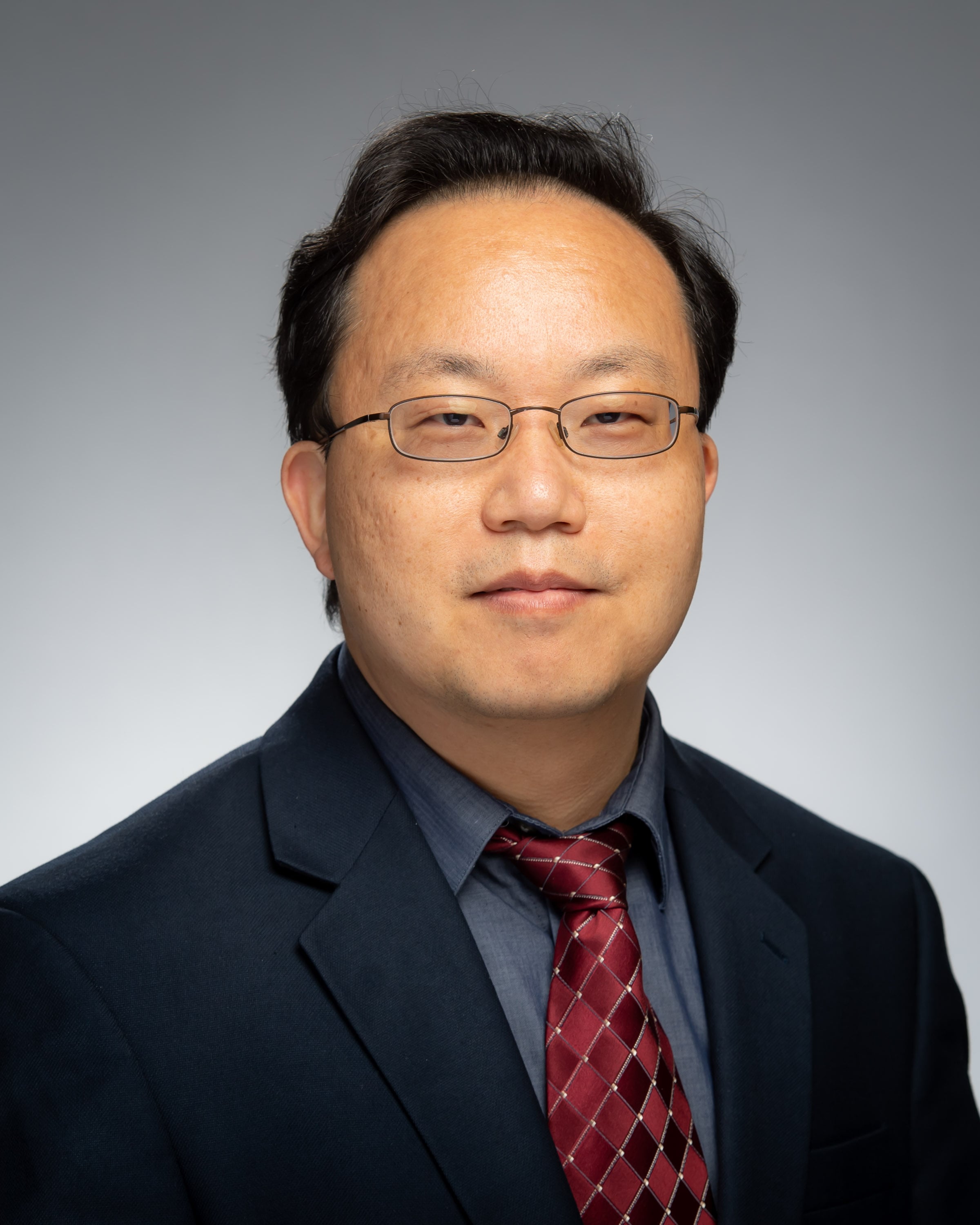}}]{Hai Lin} (Senior Member, IEEE) 
is a professor at the Department of Electrical Engineering, University of Notre Dame, where he got his Ph.D. in 2005. Before returning to his {\em alma mater}, he worked as an assistant professor in the National University of Singapore from 2006 to 2011. Dr. Lin's teaching and research activities focus on the multidisciplinary study of fundamental problems at the intersections of control theory, machine learning and formal methods. His current research thrust is motivated by challenges in cyber-physical systems, long-term autonomy, multi-robot cooperative tasking, and human-machine collaboration. Dr. Lin has served on several committees and editorial boards, including {\it IEEE Transactions on Automatic Control}. He served as the chair for the IEEE CSS Technical Committee on Discrete Event Systems from 2016 to 2018, the program chair for IEEE ICCA 2011, IEEE CIS 2011 and the chair for IEEE Systems, Man and Cybernetics Singapore Chapter for 2009 and 2010. He is a senior member of IEEE and a recipient of 2013 NSF CAREER award.
\end{IEEEbiography}

\begin{IEEEbiography}[{\includegraphics[width=1in,height=1.25in,clip,keepaspectratio]{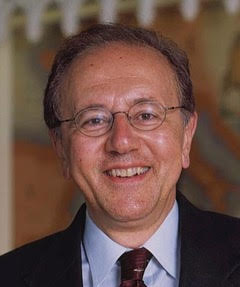}}]{Panos J. Antsaklis} (Fellow, IEEE) 
is the H.C. \& E.A. Brosey Professor of Electrical Engineering at the University of Notre Dame. He is graduate of the National Technical University of Athens, Greece, and holds MS and PhD degrees from Brown University. His research addresses problems of control and automation and examines ways to design control systems that will exhibit high degree of autonomy. His current research focuses on Cyber-Physical Systems and the interdisciplinary research area of control, computing and communication networks, and on hybrid and discrete event dynamical systems. He is IEEE, IFAC and AAAS Fellow, President of the Mediterranean Control Association, the 2006 recipient of the Engineering Alumni Medal of Brown University and holds an Honorary Doctorate from the University of Lorraine in France. He served as the President of the IEEE Control Systems Society in 1997 and was the Editor-in-Chief of the IEEE Transactions on Automatic Control for 8 years, 2010-2017.
\end{IEEEbiography}

\end{document}